\newif\ifdraft \draftfalse

\newif \ifsubmission \submissiontrue
\newif \ifshort \shorttrue
\newif \iffirst \firsttrue

\submissionfalse
\draftfalse

\makeatletter \@input{tex.flags} \makeatother

\documentclass[11pt]{article}

\newcommand{\SUBSECTION}[1]{\ifshort\paragraph*{#1}\else\subsection{#1}\fi}

\newenvironment{mathdisplayfull}{\ifshort $ \else \begin{displaymath} \fi}{\ifshort
  $ \ignorespaces \else \end{displaymath}\fi}

\newenvironment{smalldisplay}{\ifshort \begin{center} $ \else
  \begin{displaymath}\fi}
  {\ifshort $ \end{center} \else \end{displaymath} \fi}

\newcommand{\longquad}{{\ifshort \;\; \else  \qquad \fi}}

\newcommand{\thelongref}[1]{\ifshort the extended version\else \Cref{#1}\fi}
\newcommand{\preflong}[1]{\ifshort\unskip\else(\Cref{#1})\fi}

\usepackage{algorithm}
\usepackage[noend]{algorithmic}

\usepackage[margin=1in]{geometry}
\usepackage{amsmath, amssymb, amsthm}
\usepackage{mathtools}
\usepackage{verbatim}
\usepackage{color}
\usepackage{xcolor}

\usepackage{multicol}

\definecolor{DarkGreen}{rgb}{0.1,0.5,0.1}
\definecolor{DarkRed}{rgb}{0.5,0.1,0.1}
\definecolor{DarkBlue}{rgb}{0.1,0.1,0.5}
\usepackage[]{hyperref}
\hypersetup{
    unicode=false,          
    pdftoolbar=true,        
    pdfmenubar=true,        
    pdffitwindow=false,      
    pdftitle={},    
    pdfauthor={}
    pdfsubject={},   
    pdfnewwindow=true,      
    pdfkeywords={keywords}, 
    colorlinks=true,       
    linkcolor=DarkRed,          
    citecolor=DarkGreen,        
    filecolor=DarkRed,      
    urlcolor=DarkBlue,          
}
\usepackage{xspace}
\usepackage{cleveref}
\usepackage{thmtools, thm-restate}
\usepackage{mathpartir}

\ifsubmission
\usepackage[numbers,sort&compress]{natbib}
\else
\usepackage{natbib}
\fi

\newcommand{\sw}[1]{\ifdraft \textcolor{blue}{[Steven: #1]}\fi}
\newcommand{\jh}[1]{\ifdraft \textcolor{red}{[Justin: #1]}\fi}
\newcommand{\ar}[1]{\ifdraft \textcolor{green}{[Aaron: #1]}\fi}

\newcommand\RR{\mathbb{R}}
\newcommand\EE{\mathbb{E}}
\newcommand\cA{\mathcal{A}}
\newcommand\cB{\mathcal{B}}

\newcommand\cL{\mathcal{L}}

\newcommand\cM{\mathcal{M}}
\newcommand\cO{\mathcal{O}}
\newcommand\cP{\mathcal{P}}
\newcommand\cQ{\mathcal{Q}}
\newcommand\cR{\mathcal{R}}
\newcommand\cX{\mathcal{X}}
\newcommand\cD{\mathcal{D}}
\newcommand\cU{\mathcal{U}}
\newcommand\cN{\mathcal{N}}

\newcommand\cY{\mathcal{Y}}

\renewcommand{\tilde}{\widetilde}

\DeclareMathOperator*{\Expectation}{\mathbb{E}}
\newcommand{\Ex}[2]{\Expectation_{#1}\left[#2\right]}

\newcommand{\PDD}{\mbox{{\sf PrivDude}}\xspace}
\newcommand{\TPDD}{\mbox{{\sf TrueDude}}\xspace}
\newcommand{\TRPDD}{\mbox{{\sf TightDude}}\xspace}
\newcommand{\RPDD}{\mbox{{\sf RoundDude}}\xspace}

\newcommand{\eps}{\varepsilon}
\def\epsilon{\varepsilon}

\DeclareMathOperator{\Lap}{Lap}
\newcommand{\OPT}{\mathrm{OPT}}
\renewcommand{\hat}{\widehat}
\renewcommand{\bar}{\overline}

\DeclareMathOperator*{\argmax}{\mathrm{argmax}}

\newcommand{\INDSTATE}[1][1]{\STATE\hspace{#1\algorithmicindent}}

\def\<{\left\langle}
\def\>{\right\rangle}

\newtheorem*{remark*}{Remark}
\newtheorem*{corollary*}{Corollary}
\newtheorem*{theorem*}{Theorem}
\newtheorem*{proposition*}{Proposition}
\newtheorem*{assumption*}{Assumption}
\newtheorem*{example*}{Electricity Example}
\declaretheorem[
  name=Theorem,
  refname={theorem, theorems},
  Refname={Theorem, Theorems}]{theorem}
\declaretheorem[
  name=Lemma,
  refname={lemma, lemmas},
  Refname={Lemma, Lemmas}]{lemma}
\declaretheorem[
  name=Fact,
  refname={fact, facts},
  Refname={Fact, Facts}]{fact}

\declaretheorem[
  name=Corollary,
  refname={corollary, corollaries},
  Refname={Corollary, Corollaries}]{corollary}
\declaretheorem[
  name=Definition,
  refname={definition, definitions},
  Refname={Definition, Definitions}]{definition}

\title{Jointly Private Convex Programming}

\author{Justin Hsu\thanks{Department of Computer and Information
    Science, University of Pennsylvania. Email: {\tt
      justhsu@cis.upenn.edu}} \qquad Zhiyi Huang\thanks{Department of
    Computer Science, University of
    Hong Kong. Email: {\tt zhiyi@cs.hku.hk}} \qquad Aaron
  Roth\thanks{Department of Computer and Information Science,
    University of Pennsylvania. Email: {\tt aaroth@cis.upenn.edu}}
  \qquad Zhiwei Steven Wu\thanks{Department of Computer and Information Science,
    University of Pennsylvania. Email: {\tt wuzhiwei@cis.upenn.edu}}
}

\begin{document}

\maketitle

\begin{abstract}

In this paper we present an extremely general method for approximately solving a
large family of convex programs where the solution can be divided between
different agents, subject to joint differential privacy. This class includes
multi-commodity flow problems, general allocation problems, and
multi-dimensional knapsack problems, among other examples.
The accuracy of our algorithm depends on the \emph{number} of constraints that
bind between individuals, but crucially, is \emph{nearly independent} of the
number of primal variables and hence the number of agents who make up the
problem. As the number of agents in a problem grows, the error we introduce
often becomes negligible.

We also consider the setting where agents are strategic and have preferences
over their part of the solution. For any convex program in this class that
maximizes \emph{social welfare}, there is a generic reduction that makes the
corresponding optimization \emph{approximately dominant strategy truthful} by
charging agents prices for resources as a function of the
approximately optimal dual variables, which are themselves computed under
differential privacy.
Our results substantially expand the class of problems that are known to be
solvable under both privacy and incentive constraints.
\end{abstract}

\vfill
\thispagestyle{empty}
\setcounter{page}{0}
\pagebreak

\ifsubmission
  \shorttrue
\else
  \shortfalse
\fi

\firsttrue

\section{Introduction}

One hot summer, a large electricity provider has a big problem: electricity
demand is in
danger of rising above generation capacity. Rather than resorting to rolling
brown-outs, however, the utility has the ability to remotely modulate the
air-conditioners of individual buildings. The utility hopes that it can
coordinate shut-offs so that nobody is ever uncomfortable---say, guaranteeing
that every apartment's air conditioner runs during some 10 minute interval of
every hour when someone is home---while keeping peak electricity usage under
control.

The utility's scheduling problem can be modeled as a capacitated max-weight
matching problem with a minimum service constraint: we want to match apartments
to 10-minute intervals of time such that a) the usage is below the max power
production capacity at all time, and b) every apartment is matched to at least
one interval in every hour in which it is occupied. However, there is a privacy
concern: the utility will now make decisions \emph{as a function of when
  customers report they are at home}, which is sensitive information. How can we
solve this problem so that no coalition of customers $j \neq i$ can learn about
the schedule of customer $i$?

There is also a question of incentives: customers may lie about when they are
home if doing so gives them better access to air-conditioning without requiring
that they pay more. Can we solve the scheduling problem while also producing
electricity prices (on an interval by interval basis) so that no customer
has significant incentive to misreport their values?  Since the prices will be
publicly available, they must also protect the privacy of the customer
schedules.

Thinking more generally, we can imagine many other similar distributed, private
optimization problems:

\begin{itemize}
  \item \emph{Privacy-preserving multi-commodity flow:~} Consider the
      problem of routing packets through a network to minimize the total
      cost subject to edge capacities. Here, the private data of each
      individual corresponds to her source-destination pair, which may be
      sensitive information. We would like to solve this problem so that no
      coalition of other individuals $j \neq i$ can learn about the
      source/destination pair of agent $i$. We also want to produce a
      price per edge so that no agent has significant incentive to misreport
      the valuations or source/destination pairs.

  \item \emph{Privacy-preserving multi-dimensional knapsack:~} Consider the
    problem of scheduling jobs on a cluster that has constraints on various
    kinds of resources (CPU cycles, disk space, RAM, GPU cycles, etc.).
    The jobs have different resource demands and values.  We wants to schedule
    the most valuable set of jobs subject to the resource constraints, without
    revealing potentially sensitive information about any individual agent's
    jobs to the other agents. We also want to set prices for each
    resource so that no agent has a significant incentive to misreport the value
    or resource demands of their jobs.
\end{itemize}

To solve these problems (and many others---see \Cref{sec:examples}) in one
stroke, we exploit a common structure of these problems: they are linear
programs (or more generally, convex programs) where both the input data and the
output are partitioned among the agents defining the problem. The algorithms
need to report to each agent her air conditioning schedule, or her route through
the network, or whether her job is scheduled on the cluster, but not the other
agent's portion of the output. An appropriate notion of privacy in this setting
is \emph{joint differential privacy} \citep{KPRU14}.

Informally, joint differential privacy requires the joint distribution on the
portion of the outputs given to agents $j \neq i$ to be insensitive to the
portion of the input provided by agent $i$---the ``combined view'' of the other
agents should not reveal much about the remaining agent. In contrast, the
standard differential privacy guarantee requires the \emph{entire} output must
be insensitive to any agent's input. Jointly differentially private algorithms
are known for a few problem (e.g., vertex cover and set cover \citep{GLMRT10},
synthetic data generation for data analysts who want privacy for their queries
\citep{DNV12, HRU13}, equilibrium computation in certain large games
\citep{KPRU14, RR14, CKRW14}), but these algorithms are custom-designed for each
problem, and do not easily generalize.

\subsection{Our Results and Techniques}

Our main contribution is a general technique for solving a large family of
convex optimization problems under joint differential privacy. Concretely,
consider any convex optimization problem that can be written in the following
\emph{linearly separable} form:
\[
\begin{aligned}
& \underset{x}{\textrm{maximize}} & & \textstyle \sum_{i=1}^n v^{(i)} (x^{(i)})
+ v^{(0)}(x^{(0)}) & & \\
& \textrm{subject to} & & x^{(i)} \in S^{(i)} & & (\text{for all } i) \\
& & & \textstyle \sum_{i=1}^n c_j^{(i)} (x^{(i)}) +c_j^{(0)} (x^{(0)}) \leq b_j
& & (\text{for all } j) .
\end{aligned}
\]

\noindent Here, $x^{(i)}$ denotes the variables that form agent $i$'s portion
of the output, for $i = 1, \dots, n$. We also allow data for a special ``agent
0'' to model auxiliary variables and constraints that don't depend on private
data. The functions $\{v^{(i)}\}_{i}$ in the objective are concave, while the
constraint functions $\{c_j^{(i)}\}_{i,j}$ are convex; both can depend on the
private data of agent $i$. The compact and convex sets $S^{(i)}$ model the
feasible region for a single agent; they naturally depend on the private data.

There are two types of constraints in the above convex program. The first
type of constraints ($x^{(i)} \in S^{(i)}$) involve \emph{only} the variables
of a single agent (for example, the flow conservation constraints between an
agent's private source and destination in the multi-commodity flow problem).
These are the ``easy'' constraints from a privacy perspective---if these were
the only kinds of constraints, then each agent could separately optimize her
portion of the objective subject to these easy constraints. This is trivially
jointly differentially private: an agent's output would be independent of the
other agents' data.


The second type of constraints involve variables from different agents (for
example, the capacity constraints in multi-commodity flow, or resource
constraints in multi-dimensional knapsack). These constraints are the
``hard'' constraints from the perspective of privacy in the sense that they
are the ones that couple different agents and require that the problem be
solved in a coordinated manner.

We give a general method for solving such problems so that we
\emph{approximately} optimize the objective and \emph{approximately} satisfy
the coupling constraints, while exactly satisfying the personal constraints,
all subject to joint differential privacy:

\begin{theorem*}[Informal, some important parameters missing]
  There is an $\epsilon$-jointly differentially private algorithm which can find
  a solution to linearly separable convex programs (1) exactly satisfying the
  personal constraints, (2) obtaining objective value at least $\OPT - \alpha$,
  and (3) guaranteeing that the sum violation over the $k$ coupling constraints
  is at most:
  \[
    \alpha \approx \tilde{O}\left(\frac{k\sigma}{\epsilon}\right)
  \]
  where $\sigma$ is a measure of the sensitivity of the convex program. For
  packing problems, we can also guarantee \emph{no} violation of the coupling
  constraints at a small additional cost in the objective. For a broad class of
  problems, we can also round to integral solutions with little additional loss.
\end{theorem*}

Our technique is based on the \emph{dual decomposition} method in distributed
optimization.\footnote{%
  See \citet{dualdecomp} for an overview of this and related techniques.}
While the original motivation for such techniques was to solve large
optimization problems on distributed networks of computers, distributed solvers
are conceptually attractive from a privacy standpoint: they give a way to solve
problems in which data is distributed, while minimizing the amount of
communication necessary between players.

Specifically, we construct a ``partial Lagrangian'' by bringing the ``hard''
constraints into the objective, while letting the ``easy'' constraints continue
to constrain the primal feasible region. Doing this leaves us with the following
equivalent minimax problem:
\[
  \textstyle
  \underset{\lambda \ge 0}{\textrm{minimize}} \,\,\, \underset{x : x^{(i)} \in S^{(i)}}{\textrm{maximize}} \,\,\, \cL(x, \lambda) := \sum_{i} v^{(i)} (x^{(i)}) - \sum_j \lambda_j \left(\sum_i c_j^{(i)}(x_i) - b_j\right) .
\]
The partial Lagrangian can be viewed as the payoff of zero-sum game between the
primal player (the maximization player, controlling the $x$ variables) and the
dual player (the minimization player, controlling the $\lambda$ variables).
Equilibrium strategies $(x^*, \lambda^*)$ form the optimal primal and dual
solutions. Further, approximate equilibria correspond to approximately optimal
primal-dual solution pairs (in a sense we will make precise).

We then compute an approximate equilibrium of the above zero-sum game
under joint differential privacy. Without privacy, there are a variety
of methods to do this; we do the following: Repeatedly simulate play
of the game, letting one player (for us, the dual player) update his
variables using a no-regret algorithm (for us, gradient descent),
while letting the other player (for us, the primal player) repeatedly best
respond to his opponent's actions. The time averaged play of this
simulation is known to converge quickly to an approximate minimax
equilibrium of the game~\citep{FS96}.

In our case, the best-response problem for the primal player has a particularly
nice form. Crucially, the form of the Lagrangian allows us to compute a primal
best response by having each of the $n$ agents best respond \emph{individually},
without coordination: Given the current dual variables $\lambda^{(t)}_j$, agent
$i$ solves the problem
\[
\begin{aligned}
& \underset{x^{(i)}}{\textrm{maximize}} & & \textstyle v^{(i)} (x^{(i)}) - \sum_j \lambda^{(t)}_j c_j^{(i)}(x^{(i)}) \\
& \text{subject to} & & x^{(i)} \in S^{(i)} ,
\end{aligned}
\]
which is an optimization problem that is independent of the private data of
other players $j \neq i$. The combination of these individual agent best
responses forms a primal best response in the optimization problem, and the
time-averaged strategies of individual agents will form our near optimal primal
solution. Because the portion of the solution corresponding to agent $i$ was
generated in a method that was \emph{independent} of the private data of any
other agent $j \neq i$ (except through the dual variables), the primal solution
we compute will satisfy joint-differential privacy so long as the sequence of
plays computed by the dual player satisfies differential privacy.

All that remains is to give a privacy preserving implementation of the dual
player's updates. To do this, we use relatively standard techniques to implement
privacy preserving gradient descent, by adding Gaussian noise to the gradient of
the Lagrangian at each step.

Finally, consider the setting where the agents have preferences over their parts
of the final solution, and the objective of the convex program is to maximize
the total value of all agents (\emph{social welfare}). We can guarantee
asymptotic truthfulness at little additional cost to the approximation factor,
with only minor modifications to the algorithm.  By taking our approximately
optimal primal/dual solution pair and making a small modification (losing a bit
more in approximate feasibility), we can treat the dual variables as prices.
Then, every agent is allocated her favorite set of primal variables given her
constraints and the prices.\footnote{%
  The allocation and prices can also be seen as an \emph{Walrasian equilibrium}
  \citep{GS99}.}

In general, this property would not be enough to guarantee truthfulness: while
no agent would prefer a different solution at the given prices, agents may be
able to \emph{manipulate} the prices to their advantage (say, by lowering the
price on the goods they want).  However, in our case, the prices arise from the
dual variables, which are computed under differential privacy; agents can't
substantially manipulate them. As we show, this makes truth telling an
approximate dominant strategy for all players. In the ``large game'' setting
where the number of players $n$ increases, the amount any player can gain by
deviating will tend to $0$ if the constraints are also relaxed.


\subsection{Related Work}


Distributed optimization techniques date back to the 1950s, with the original
goal of solving large optimization problems with networks of computationally
limited machines. The area is very rich with mathematically elegant algorithms,
many of which work well in practice. 
The method we develop in this paper is based on a simple technique called dual
decomposition. For a more comprehensive overview, the reader can consult the
excellent survey by \citet{dualdecomp}.


Differential privacy emerged from a series of papers \citep{DN03,BDMN05,DMNS06},
culminating in the definition given by \citet{DMNS06}. This literature is far
too vast to summarize, but we refer the reader to \citet{DR14} for a textbook
introduction.

This paper fits into a line of work which seeks to solve problems (mostly, but
not exclusively optimization problems) which cannot be solved under the standard
constraint of differential privacy, but must be solved under a relaxation
(termed \emph{joint differential privacy} by \citet{KPRU14}) that allows the
portion of the solution given to agent $i$ to be sensitive in agent $i$'s data.
Our work is most related to \citet{HHRRW14}, who give an algorithm for
computing approximately max-weight matchings under joint differential
privacy.  This algorithm implements an ascending price auction using a
deferred-acceptance-esque algorithm and also uses prices (which are the dual
variables in the matching problem) to coordinate the allocation.
\citet{HHRRW14} use techniques rather specific to matchings, which do not
generalize even to the $k$-demand allocation problem (when agents have
general preferences, rather than ``gross substitute preferences''), and
certainly not to general convex optimization problems. In fact, simply
solving the allocation problem beyond gross substitutes valuations was stated
as the main open problem by \citet{HHRRW14}; our technique in particular
solves this problem as a special case. Our algorithm also gives approximate
truthfulness, unlike the one by \citet{HHRRW14}, which had to sacrifice
approximate truthfulness to improve the accuracy of their algorithm.

For related work in private optimization, \citet{HRRU14} consider how to solve
various classes of linear programs under (standard) differential privacy.  Their
work contains many negative results because many natural linear programs cannot
be solved under the standard constraint of differential privacy, while we give
broadly positive results (under the looser notion of joint differential
privacy).  \citet{NRS07} also consider some combinatorial optimization problems,
also under the standard constraint of differential privacy.

Our work is also related to the literature on solving private empirical risk
minimization problems, which involves \emph{unconstrained} convex minimization
subject to the standard differential privacy constraint
\citep{CMS11,JKT11,KST12,JT14,BST14}.  Many of these papers use privacy
preserving variants of gradient descent (and other optimization algorithms),
which we use as part of our algorithm  (for the \emph{dual} player).

The connection between differential privacy and approximate truthfulness was
first made by \citet{MT07}, and later extended in many papers (e.g., by
\citet{NST12} and \citet{CCKMV13}).  \citet{HK12} showed a generic (and
typically computationally inefficient) method to make welfare maximization
problems exactly truthful subject to (standard) differential privacy.  The
method of \citet{HK12}, and more generally any differentially private
algorithm, cannot solve most natural auction problems (in which each agent
receives some allocation) with non-trivial social welfare guarantees.
In contrast, we show a generic method of obtaining \emph{joint} differential
privacy and (approximate) truthfulness for any problem that can be posed as a
linearly separable convex program, which covers (the relaxations of)
combinatorial auctions as a special case.  Our method is also efficient
whenever the non-private version can be solved efficiently.

The connection between joint differential privacy and truthfulness is
more subtle than the connection with standard differential privacy. In order for a jointly differentially private
computation to be approximately truthful, it must in some sense be
computing an equilibrium of an underlying game \citep{KPRU14,RR14,
  CKRW14}. In our case, this corresponds to computing dual variables
which serve as Walrasian-equilibrium like prices. This also
circumvents the criticisms of (standard) differential privacy as a
solution concept, given by \citet{NST12,Xia13}---under joint
differential privacy, it is not the case that \emph{all} strategies
are approximate dominant strategies, even though truthful reporting
can be made so.

\section{Preliminaries}

\SUBSECTION{Differential Privacy Preliminaries} We represent the private data
of an agent as an element from some \emph{data universe} $\cU$. A
\emph{database} $D\in \cU^n$ is a collection of private data from $n$ agents.
Two databases $D, D'$ are \emph{$i$-neighbors} if they differ only in their
$i$-th index: that is, if $D_j = D_j'$ for all $j\neq i$. If two databases
$D$ and $D'$ are $i$-neighbors for some $i$, we say they are
\emph{neighboring
  databases}. We will be interested in randomized algorithms that
take a databases as input, and output an element from some abstract range $\cQ$.

\begin{definition}[\citet{DMNS06}]
A mechanism $\cM\colon \cU^n \rightarrow \cQ$ is $(\eps,
\delta)$-\emph{differentially private} if for every pair of
neighboring databases $D, D'\in \cU^n$ and for every subset of outputs
$S\subseteq \cQ$,
\[
\Pr[\cM(D) \in S] \leq \exp(\eps) \Pr[\cM(D') \in S] + \delta.
\]
\end{definition}

For the class of problems we consider, elements in both the domain and the
range of the mechanism are partitioned into $n$ components, one for each
agent. In this setting, \emph{joint differential privacy}~\citep{KPRU14} is a
more natural constraint: For all $i$, the \emph{joint} distribution on
outputs given to players $j\neq i$ is differentially private in the input of
agent $i$. Given a vector $x = (x_1, \ldots, x_n)$, we write $x_{-i} = (x_1,
\ldots, x_{i-1}, x_{i+1} , \ldots, x_n)$ to denote the vector of length
$(n-1)$ which contains all coordinates of $x$ except the $i$-th coordinate.

\begin{definition}[\citet{KPRU14}]
A mechanism $\cM \colon \cU^n \rightarrow \cQ^n$ is $(\eps,
\delta)$-\emph{jointly differentially private} if for every $i$, for
every pair of $i$-neighbors $D, D'\in \cU^n$, and for every subset of
outputs $S\subseteq \cQ^{n-1}$,
\[
\Pr[\cM(D)_{-i} \in S] \leq \exp(\eps) \Pr[\cM(D')_{-i} \in S] + \delta.
\]
\end{definition}
Note that this is still a very strong privacy guarantee; the mechanism preserves
the privacy of any agent $i$ against arbitrary coalitions of other agents. It
only weakens the constraint of differential privacy by allowing agent $i$'s
output to depend arbitrarily on her \emph{own} input.


We rely on a basic but useful way of guaranteeing joint differential
privacy---the \emph{billboard model}. Algorithms in the billboard model compute
some differentially private signal (which can be viewed as being visible on a
public billboard); then each agent $i$'s output is computed as a function only
of this private signal and the private data of agent $i$. The following lemma
shows that algorithms operating in the billboard model satisfy joint
differential privacy.

\begin{lemma}[\citet{HHRRW14, RR14}] \label{billboard}
  Suppose $\cM : \cU^n \rightarrow \cQ$ is $(\eps,
  \delta)$-differentially private. Consider any set of functions $f_i
  : \cU \times \cQ \rightarrow \cQ'$. Then the mechanism $\cM'$ that
  outputs to each agent $i$: $f_i(D_i, \cM(D))$ is $(\eps,
  \delta)$-jointly differentially private.
\end{lemma}

\ifshort
We defer other privacy preliminaries and no-regret preliminaries to the full
version.
\else
We will also use the \emph{advanced composition theorem}, which shows how our
differential privacy parameters degrade under adaptive composition of
differentially private subroutines.
\begin{lemma}[\citet{dwork-composition}]\label{lem:composition}
Let $\cA\colon \cU^n \rightarrow \cQ^T$ be a $T$-fold adaptive
composition of $(\eps, \delta)$ differentially private mechanisms.
Then $\cA$ satisfies $(\eps', T\delta + \delta')$-differential privacy
for \[
\eps' = \eps \sqrt{2T\ln(1/\delta')} + T\eps(e^\eps - 1).
\]
In particular, for any $\eps \leq 1$, if $\cA$ is a $T$-fold adaptive
composition\footnote{%
  See~\citet{dwork-composition} for a more detailed
  discussion of $T$-fold adaptive composition.}
of $(\eps/\sqrt{8T\ln(1/\delta')} , \delta)$-differentially private
mechanisms, then $\cA$ satisfies $(\eps, T\delta +
\delta')$-differential privacy.
\end{lemma}
Another basic tool we need is the
\emph{Gaussian Mechanism}, which releases private perturbations of vector valued
functions $f:\cU^n\rightarrow \mathbb{R}^k$ by adding Gaussian noise with scale
proportional to the $\ell_2$ sensitivity of the function $f$:
\[
\Delta_2(f) = \max_{\text{neighboring }D, D'\in \cU^n} \|f(D) -
f(D')\|_2.
\]
\begin{lemma}[\citet{DMNS06} (see e.g. \citet{DR14} for a proof)]
  \label{lem:l2-noise}
  Let $\eps, \delta \in (0,1)$ be arbitrary, $f\colon \cU^n
  \rightarrow \RR^k$ be a function with $\ell_2$-sensitivity
  $\Delta_2(f)$. Let $c$ be a number such that $c^2 \geq
  2\log(1.25/\delta)$, then the Gaussian Mechanism (which outputs $f(D) + Z$
  where $Z \sim \cN(0, \sigma^2)^k$ with $\sigma \geq c \Delta_2(f)/\eps$) is
  $(\eps, \delta)$-differentially private.
\end{lemma}

\SUBSECTION{Zero-Sum Games and No-Regret Dynamics}
Consider a two-player zero-sum game with payoff function $A:\cX \times \cY
\rightarrow \mathbb{R}$. The maximization player selects an action $x \in \cX$,
with the goal of maximizing $A(x,y)$. Simultaneously, the minimization player
selects an action $y \in \cY$ with the goal of minimizing $A(x,y)$.

We will consider games such that:
\begin{itemize}
  \item for all $x,y \in \cX\times\cY$ the payoff $A(x,y)$ is bounded;
  \item $\cX$ and $\cY$ are 
    closed, compact, convex sets;
  \item fixing $y \in \cY$, $A(x,y)$ is a concave function in $x$; and
  \item fixing $x \in \cX$, $A(x,y)$ is a convex function in $y$.
\end{itemize}

It is known that any such game has a \emph{value} $V$: the maximization
player has an action $x^* \in \cX$ such that $A(x^*, y) \geq V$ for all $y \in
\cY$, and the minimization player has action $y^*$ such that $A(x,y^*) \leq V$
for all $x \in \cX$~\citep{Kneser52}. We can define \emph{approximate} minimax
equilibria of such games as follows:

\begin{definition}
Let $\alpha \geq 0$. A pair of strategies $x\in \cX$ and $y\in \cY$ form
an $\alpha$-\emph{approximate minimax equilibrium} if
\[
A(x, y) \geq \max_{x'\in \cX} A(x', y) - \alpha \geq V - \alpha \quad \mbox{ and } \quad A(x,y)
\leq \min_{y'\in \cY} A(x, y') + \alpha \leq V + \alpha.
\]
\end{definition}
To compute an approximate minimax equilibrium, we will use the
following $T$-round no-regret dynamics: one player plays online
gradient descent~\citep{zinkevich} as an no-regret algorithm, while
the other player selects a \emph{best-response} action in each round.
We will let the min player be the no-regret learner, who produces a sequence of
actions $\{y_1, \ldots, y_T\}$ against max player's best responses $\{x_1,
\ldots, x_T\}$. For each $t\in[T]$,
\[
y_{t+1} =  \Pi_\cY \left[ y_t - \eta\cdot {\nabla_{y} A(x_t, y_t)}\right]
\qquad \mbox{ and } \qquad
x_{t} = \arg\max_x A(x, y_t),
\]
where $\Pi_\cY$ is the Euclidean projection map onto the set $\cY$:
$\min_{y'} \|y - y'\|_2$, $\eta$ is the step size.
In the end, denote the minimization player's regret as
\[
\cR_y \equiv \frac{1}{T}\sum_{t=1}^T A(x_t, y_t) - \frac{1}{T}\min_{y \in \cY}
\sum_{t=1}^T A(x_t, y) .
\]

Now consider the average actions for both players in this
dynamics: $\bar x = \frac{1}{T} \sum_{t=1}^T x_t$ and $\bar y =
\frac{1}{T} \sum_{t=1}^T y_t$. \citet{FS96} showed that the average
plays form an approximate equilibrium:
\begin{theorem}[\citet{FS96}]
 \label{thm:equil-restr}
 $(\bar{x}, \bar{y})$ forms a $\cR_y$-approximate minimax equilibrium.
\end{theorem}


In order to construct the approximate equilibrium privately, we also
develop a noisy and private version of the online gradient descent
algorithm. We defer details to~\Cref{sec:polo}.
\fi

\section{Private Dual Decomposition} \label{sec:general}

\ifshort\else
Let's consider the electricity scheduling problem, which will be our running
example as we present our algorithm. Suppose we have $n$ agents, who need power
for $T$ intervals. Each interval is subdivided into $Q$ slots, and agents have
different valuations $v \in [0, 1]$ for different slots. There is a maximum
amount of electricity $c$ available for each (interval, slot) pair. Finally,
agents demand some total amount of electricity $d_t \in [0, Q]$ during each
interval,
and at most $d_{max}$ in total over all intervals. The demand may be zero for
some intervals, say, if the agent is not home. Translating the description in
the introduction, we have the following linear program:
\[
\begin{aligned}
  & \textit{\rm{maximize}} & & \sum_{i = 1}^n \sum_{t = 1}^T \sum_{q = 1}^Q
  v^{(i)}_{tq} x^{(i)}_{tq} & & \\
  & \textit{\rm{subject to}} & & \sum_{i = 1}^n x^{(i)}_{tq} \leq c_{tq} & &
  (\textrm{for } t \in [T], q \in [Q]) \\
  & & & \sum_{q = 1}^Q x^{(i)}_{tq} \geq d^{(i)}_{t},
  \quad
  \sum_{q = 1}^Q \sum_{t = 1}^T x^{(i)}_{tq} \leq d_{max},
  \quad
  \text{and} \quad x^{(i)}_{tq} \in [0, 1]
  & & (\textrm{for } i \in [n], t \in [T], q \in [Q])\,.
\end{aligned}
\]%


We consider each agent's valuations $v^{(i)}$ and demands $d^{(i)}$ to be
private. Agent $i$ will receive variables $x^{(i)}$, which indicate when she is
getting electricity. With this in mind, notice that this LP has a particularly
nice structure: the objective and first constraints are sums of terms that (term
by term) depend only a single agent's data and variables, while the second
constraints only constrain a single agent's data.  This LP is an instance of
what we call \emph{linearly separable convex programs}.

\begin{definition} \label{def:sep-convex}
  Let the data universe be $\cU$, and suppose there are $n$ individuals. We
  map each database $D \in \cU^n$ to a \emph{linearly separable convex optimization
    problem} $\cO$, which consists of the following data: for each agent $i$,
  \begin{itemize}
    \item a compact convex set $S^{(i)} \subseteq \RR^l$,
    \item a concave objective function $v^{(i)} : S^{(i)} \rightarrow \RR$,
    \item and $k$ convex constraint functions $c^{(i)}_j : S^{(i)}
    \rightarrow \RR$ (indexed by $j = 1 \dots k$);
  \end{itemize}
  all defined by $D_i$. $\cO$ also includes the following data, independent of
  D:
  \begin{itemize}
    \item a compact convex set $S^{(0)} \subseteq \RR^l$,
    \item a concave objective function $v^{(0)} : S^{(0)} \rightarrow \RR$,
    \item $k$ convex constraint functions $c^{(0)}_j : S^{(0)}
      \rightarrow \RR$ (indexed by $j = 1 \dots k$),
    \item and a vector $b \in \RR^k$.
  \end{itemize}
  The convex optimization problem is:
  \[
  \begin{aligned}
    & \textit{\rm{maximize}} & & \sum_{i = 0}^n v^{(i)} (x^{(i)}) & & \\
    & \textit{\rm{subject to}} & & \sum_{i = 0}^n c^{(i)}_j(x^{(i)}) \leq b_j &
    & (\textrm{for } j = 1 \dots k) \\
    & & & x^{(i)} \in S^{(i)} & & (\textrm{for } i = 0 \dots n) .
  \end{aligned}
  \]
  We call the set of all such optimization problems the \emph{class of problems
    associated to $\cU$} or simply a \emph{class of problems}, if $\cU$ is
  unimportant.
\end{definition}
Intuitively, the variables, objective, and constraints indexed by $i$ belong to
agent $i$ for $i \in \{1,\ldots,n\}$.  The constraints and variables for $i = 0$
are shared---it is sometimes useful to have auxiliary variables in a convex
program that do not correspond to any player's part of the solution. We call
the first kind of constraints the \emph{coupling constraints}, since they
involve multiple agents' variables and data.  We call the second kind of
constraints the \emph{personal constraints}, since they only involve a single
agent's variables and data.  We write $S = S^{(0)} \times \dots \times S^{(n)}$
for the portion of the feasible region defined only by the personal constraints.

\begin{example*}
  The coupling constraints are the power supply constraints on each time slot,
  and the personal constraints are the agent's demand constraints for different
  intervals.
\end{example*}
\fi


\subsection{Algorithm}

\ifshort Our algorithm will solve linearly separable convex programs of the
sort we saw above. \fi To solve the convex program, we will work extensively
with the \emph{Lagrangian}:
\[
  \cL(x, \lambda) = \sum_{i = 0}^n v^{(i)} (x^{(i)}) -
  \sum_{j = 1}^k \lambda_j \left( \sum_{i = 0}^n c^{(i)}_j(x^{(i)}) - b_j \right) .
\]

\ifshort\else
For clarity of presentation, we will assume in the following that our
optimization problems are feasible and that the \emph{strong duality} condition
holds, i.e.,
\[
  \max_{x \in S} \min_{\lambda \in \RR^k_+} \cL(x, \lambda)
  =
  \min_{\lambda \in \RR^k_+} \max_{x \in S} \cL(x, \lambda).
\]
If our problems are not even approximately feasible, this will be easy to
detect, and the assumption that strong duality holds is without loss of
generality in our setting, since our goal is only to approximately satisfy the
coupling constraints.\footnote{%
  Strong duality is guaranteed in particular by \emph{Slater's condition}
  \citep{Slater50}: there exists some point $x \in S$ such that for all
  $j\in[k]$, $\sum_{i=0}^n c_j^{(i)}(x) < b_j$. If this is not already the case,
  it can be guaranteed by simply relaxing our constraints by a tiny amount.
  Since our solutions will already only \emph{approximately} satisfy the
  coupling constraints, this is without loss.}
\fi

%


We will interpret the Lagrangian objective as the payoff function of a zero-sum
game between a primal player (controlling the $x$ variables), and a dual player
(controlling the $\lambda$ variables). We will privately construct an
approximate equilibrium of this game by simulating repeated play of the game. At
each step, the primal player will best-respond to the dual player by finding $x$
maximizing the Lagrangian subject to the dual player's $\lambda$. The dual
player in turn will run a no-regret algorithm to update the $\lambda$ variables,
using losses defined by the primal player's choice of $x$.  By a standard result
about repeated play \preflong{thm:equil-restr}, the average of each player's
actions converges to an approximate equilibrium.

We will first detail how to privately construct this approximate equilibrium.
Then, we will show that an approximate equilibrium is an approximately optimal
primal-dual pair for the original problem; in particular, the primal player's
point will be approximately feasible and optimal. Throughout, let the problem
be $\cO = (S, v, c, b)$.
\ifshort
For lack of space, we will only be able to sketch the argument; details are
deferred to the extended version.
\fi

\ifshort\else
\begin{remark*}
  Before we begin, we want to clarify one point. We will sometimes say ``Agent
  $i$ plays \dots'' or ``Agent $i$ solves \dots''. These descriptions sound
  natural, but are slightly misleading: Our algorithms will not be online or
  interactive in any sense, and all computation is done by our algorithm, not by
  the agents. Agents will submit their private data, and will receive a single
  output. Our algorithm will \emph{simulate} the agent's behavior, be it
  selecting actions to play, or solving smaller optimization problems, or
  rounding.
\end{remark*}
\fi

Let's start with the primal player's best-response.  Rewriting the Lagrangian
and fixing $\lambda$, at each round, the primal player plays
\begin{equation} \label{eq:br}
  \argmax_{x \in S} \cL(x) = \sum_{i = 0}^n \left( v^{(i)} (x^{(i)}) -
  \sum_{j = 1}^k \lambda_j c^{(i)}_j(x^{(i)}) \right)
  + \sum_{j = 1}^k \lambda_j b_j ,
\end{equation}
which can be computed by solving
\[
   x^{(i)} \in BR^{(i)}\left(\cO, \{\lambda_j\}\right) :=  \argmax_{x \in S^{(i)}}
  \left( v^{(i)} (x) - \sum_{j = 1}^k \lambda_j c^{(i)}_j(x) \right)
\]
independently for each agent $i$, and combining the solutions to let $x :=
(x^{(1)}, \dots, x^{(n)})$.

For the dual player, we will use the online gradient descent algorithm due to
\citet{zinkevich} for a suitably chosen target set $\Lambda$, together with
noise addition to guarantee differential privacy. At each time step we feed in a
perturbed version of the loss vector $l \in \RR^k$ defined to be the gradient of
the Lagrangian with respect to $\lambda$:
\begin{mathdisplayfull}
  l_j := \frac{\partial \cL}{\partial \lambda_j} = \sum_{i = 0}^n
  c^{(i)}_j(x^{(i)}) - b_j .
\end{mathdisplayfull}
To obtain privacy properties, we further add Gaussian noise the to above gradient and update the dual variables according to the noisy gradients.

Putting everything together, our algorithm Private Dual Decomposition
($\PDD$) presented in~\Cref{alg:general} solves linearly separable convex
programs under joint differential privacy.

\begin{algorithm}
  \caption{Joint Differentially Private Convex Solver: $\PDD(\cO,
    \sigma, \tau, w, \eps, \delta, \beta)$}
  \label{alg:general}
  \begin{algorithmic}
    \STATE{{\bf Input}: Convex problem $\cO = (S, v, c, b)$ with $n$ agents and
      $k$ coupling constraints, gradient sensitivity bounded by
      $\sigma$, a dual bound $\tau$, width bounded by $w$, and privacy
      parameters $\eps > 0,
      \delta \in (0, 1),$ confidence parameter $\beta \in (0, 1)$.}
    \STATE{{\bf Initialize}:
      \begin{mathpar}
        \lambda^{(1)}_j := 0 \text{ for } j \in [k] ,
        \and T := w^2 ,
        \and \eps' := \frac{\eps \sigma}{\sqrt{8T\ln(2/\delta)}} ,
        \and \delta' := \frac{\delta}{2T} ,
        \and \eta := \frac{2\tau}{\sqrt{T}\left( w + \frac{1}{\eps'} \log
            \left( \frac{Tk}{\beta} \right) \right)} ,
        \and \Lambda := \{ \lambda \in \RR^k_+ \mid \| \lambda \|_\infty \leq 2
        \tau \} .
      \end{mathpar}
    }
    \STATE{{\bf for iteration} $t = 1 \dots T$}
    \INDSTATE[1]{{\bf for each } agent $i = 0 \dots n$}
    \INDSTATE[2]{Compute personal best response:
      \begin{mathdisplayfull}
        x^{(i)}_t := \argmax_{x \in S^{(i)}} v^{(i)}(x) - \sum_{j = 1}^k
        \lambda^{(t)}_j c^{(i)} (x) .
      \end{mathdisplayfull}
    }
    \INDSTATE[1]{{\bf for each } constraint $j = 1 \dots k$}
    \INDSTATE[2]{Compute noisy gradient:
      \begin{mathdisplayfull}
        \hat \ell^{(t)}_j := \left( \sum_{i = 0}^n c^{(i)}(x^{(i)}_t) \right) -
        b_j + \cN\left(0, \frac{2 \sigma^2 \log \left( 1.25/\delta'
            \right)}{\eps'^2} \right),
      \end{mathdisplayfull}
    }
    \INDSTATE[1]{Do gradient descent update:
    \begin{mathdisplayfull}
      \lambda^{(t+1)} := \Pi_\Lambda \left( \lambda^{(t)} + \eta \hat \ell^{(t)}
      \right) .
    \end{mathdisplayfull}
    }
    \STATE{{\bf Output}:
        $\bar{x}^{(0)} := \frac{1}{T}\sum_{t=1}^T
        x^{(0)}_t$ and $\bar{\lambda} := \frac{1}{T}
        \sum_{t=1}^T \lambda^{(t)}$ to everyone,
        $\bar{x}^{(i)} := \frac{1}{T}\sum_{t=1}^T x^{(i)}_t$ to agents $i \in [n]$ .
      }
  \end{algorithmic}
\end{algorithm}

 \subsection{Privacy}

\ifshort\else
Next, we establish the privacy properties of \PDD.
We will first argue that the dual variables $\lambda$ satisfy (standard) differential privacy, because the algorithm adds Gaussian noise to the gradients.
Then, we will argue that the primal solution satisfies joint differential
privacy, because each agent's best-response depends only on her own private data
and the dual variables.

First, we define neighboring convex problems.
\begin{definition} \label{def:neighboring}
  Let $D, D' \in \cU^n$ be two neighboring databases. We say the associated
  convex programs $\cO, \cO'$ are \emph{neighboring} problems.
\end{definition}
\fi

By looking at how much the gradient $l_j$ may change in neighboring instances,
we can determine how much noise to add to ensure differential privacy.

\begin{definition}
  A problem $\cO$ has \emph{gradient sensitivity} bounded by $\sigma$ if
  \begin{smalldisplay}
    \max \sum_{j = 1}^k \left| c^{(i)}_j(x^{(i)}) - c'^{(i)}_j(x'^{(i)})
    \right|^2 \leq \sigma^2 ,
  \end{smalldisplay}
  where the maximum is taken over agents $i$, dual variables $\{\lambda_j\}
  \subseteq \RR^k_+$, neighboring problems $\cO$ and $\cO'$,
  and
  \begin{mathdisplayfull}
    x^{(i)} \in BR^{(i)}\left(\cO, \{\lambda_j\} \right)
    \longquad \text{and} \longquad
    x'^{(i)} \in BR^{(i)}\left(\cO', \{\lambda_j\} \right) .
  \end{mathdisplayfull}
\end{definition}


\ifshort\else
\begin{example*}
  We can bound the gradient sensitivity of the electricity scheduling LP. By
  changing her private data, an agent may change her demand vector by at most
  $d_{max}$ in each of $T$ time slots. Since the coupling constraints are simply the
  total demand over all agents for each time slot, the gradient sensitivity is
  at most $\sigma = 2\sqrt{d_{max}}$.
\end{example*}

The gradient sensitivity $\sigma$ is the key parameter for guaranteeing privacy.
By definition, it is a bound on the $\ell_2$ sensitivity of the gradient vector
$l$. By~\Cref{lem:l2-noise}, releasing the noisy vector $\hat l$ by adding
independent Gaussian noise drawn from the distribution $\cN(0, 2
\sigma^2\log(1.25/\delta)/\eps^2)$ to each coordinate satisfies $(\eps,
\delta)$-(standard) differential privacy.
\fi

Thus, the following  theorem shows privacy of the dual variables in
$\PDD$.

\begin{theorem} \label{thm:dual-priv}
  Let $\epsilon > 0, \delta \in (0, 1/2)$ be given. The sequence of dual
  variables $\lambda^{(1)}, \dots, \lambda^{(T)}$ and the public variables
  $x^{(0)}_1, \dots, x^{(0)}_T$ produced by $\PDD$ satisfy $(\epsilon,
  \delta)$-differential privacy.
\end{theorem}
\ifshort\else
\begin{proof}
 By \Cref{lem:l2-noise} and \Cref{lem:composition}.
\end{proof}
\fi

To show joint differential privacy of the primal variables, note that agent
$i$'s best response function $BR^{(i)} (\cO, \{ \lambda_j \} )$ (defined in
\Cref{eq:br}) is a function of $i$'s personal data and the current dual
variables $\lambda_j$, which satisfy standard differential privacy by
\Cref{thm:dual-priv}. So, we can use the \emph{billboard lemma}
(\Cref{billboard}) to show the sequence of best-responses satisfies
\emph{joint}-differential privacy.

\begin{theorem} \label{thm:primal-jdp}
  Let $\epsilon > 0, \delta \in (0, 1/2)$ be given. Releasing the sequence of
  private variables $x^{(i)}_1 \dots x^{(i)}_T$ to agent $i$ satisfies
  $(\epsilon, \delta)$-joint differential privacy.
\end{theorem}
\ifshort\else
\begin{proof}
  By \Cref{thm:dual-priv} and \Cref{billboard}.
\end{proof}
\fi

\subsection{Accuracy}

Now, let us turn to accuracy.
\ifshort\else
We first argue that the exact equilibrium of the
game in which we restrict the dual player's strategy space corresponds to an
optimal primal-dual pair for the original game. While the original game allows
the dual variables to lie anywhere in $\RR^k_+$, we need to restrict the dual
action space to a bounded subset $\Lambda$, in order to use gradient descent. We
first show that if $\Lambda$ is a large enough set, restricting the dual player
to play in $\Lambda$ will not change the equilibrium strategy and value of the
game.  We next show that \PDD computes an approximate equilibrium
of the Lagrangian game. Finally, we show that the approximate equilibrium
strategy of the primal player must be an approximately feasible and optimal
point of the original convex program.

For the first step, we observe that the optimal primal and dual solutions $(x^*,
\lambda^*)$ achieve the value of the unrestricted game, which is the optimal
objective value $\OPT$ of the original convex program.

\begin{lemma} \label{lem:opt-value}
  Let $(x^*, \lambda^*)$ achieve
  \[
    \argmax_{\lambda \in \RR^k_+} \min_{x \in S} \cL(x, \lambda) .
  \]
  Then,
  \begin{itemize}
    \item $x^*$ is a feasible solution to the original convex program, and
    \item  $\cL(x^*, \lambda^*) = \OPT$, the optimal objective value of the original convex program.
  \end{itemize}
\end{lemma}
\begin{proof}
  Follows directly from strong duality of the convex problem.
\end{proof}

We now reason about the \emph{restricted} game, in which the dual player plays
in a subset $\Lambda \subseteq \RR^k_+$.
\fi
We first define a key parameter for the
accuracy analysis, which measures how much the objective can be improved beyond
$\OPT$ by \emph{infeasible} solutions, as a function of how much the infeasible
solution violates the constraints.

\begin{definition}
  Consider all instances $\cO = (S, v, c, b)$ in a class of problems. We call
  $\tau > 0$ a \emph{dual bound} for the class if for all $\delta \geq 0$, $i$,
  and $x \in S$ such that
  \[
    \sum_{i = 0}^n \sum_{j = 1}^k \left( c^{(i)}_j (x^{(i)}) - b_j \right)_+
    \leq \delta ,
    \qquad \text{we have} \qquad
    \sum_{i = 0}^n v^{(i)} (x^{(i)}) \leq \OPT(\cO) + \tau \delta ,
  \]
  where $\OPT(\cO)$ is the optimum objective for $\cO$. (We will frequently
  elide $\cO$, and just say $\OPT$ when the convex program is clear.)
\end{definition}

\ifshort\else
\begin{example*}
  The coupling constraints are power supply constraints. Violating these
  constraints by $\delta$ in total will increase the objective by at most
  $\delta v_{max} \leq \delta$, so $\tau = 1$ is a dual bound for this problem.
\end{example*}

Intuitively, the dual bound indicates how much the objective can increase
beyond the optimal value by slightly violating the feasibility constraints. It
will control how large the dual action space must be, in order to discourage the
primal player from playing an infeasible point at equilibrium. More precisely,
we can show that the game with dual actions restricted to $\Lambda$ still has
the optimal and feasible point as the equilibrium strategy for the primal
player, if $\Lambda$ is large enough.

\begin{lemma} \label{lem:restrict}
  Suppose $\tau$ is a dual bound for a class of convex optimization problems.
  Then if we restrict the dual action space to be $\Lambda = \{ \lambda \in
  \RR^k_+ \mid \| \lambda \|_2 \leq 2 \tau \sqrt{k} \}$, there is a dual strategy
  $\lambda^\bullet \in \Lambda$ such that $(x^*, \lambda^\bullet)$ is an
  equilibrium of the restricted game.
\end{lemma}
\begin{proof}
  By strong duality, there exists $(x^*, \lambda^*)$ an equilibrium of the
  Lagrangian game with value $\OPT$. Played against any strategy in $\Lambda$,
  $x^*$ gets value at least $\OPT$ since it is an optimal, feasible solution of
  the convex program. We first show that restricting the dual player's action
  set to $\Lambda$ leaves the value of the game at $\OPT$.

  Consider any other primal action $x$. If it doesn't violate any coupling
  constraints, then the dual player can set all dual variables to $0$. Thus, $x$
  has payoff at most $\OPT$ in the worst case over all dual player strategies in
  $\Lambda$.

  On the other hand, suppose $x$ violates some constraints:
  \[
    \delta := \sum_{j = 1}^k \left( \sum_{i = 0}^n c^{(i)}_j (x^{(i)}) - b_j
    \right)_+ > 0 .
  \]
  We can construct a dual player action $\lambda' \in \Lambda$ to give the
  primal player payoff strictly less than $\OPT$:
  \begin{enumerate}
    \item For constraints $j$ where $x$ violates the constraints
      \[
        \sum_{i = 0}^n c^{(i)}_j (x^{(i)}) > b_j ,
      \]
      set $\lambda_j' = 2\tau$.
    \item For constraints $j$ where $x$ satisfies the constraint:
      \[
        \sum_{i = 0}^n c^{(i)}_j (x^{(i)}) \leq b_j ,
      \]
      set $\lambda_j' = 0$.
  \end{enumerate}
  Note that $\lambda'$ is a valid dual action in the restricted game, since
  $\lambda' \in \Lambda$. Now, let's bound the payoff $\cL(x, \lambda')$ by
  comparing it to $\cL(x^*, \lambda^*)$. By assumption, the objective term
  increases by at most $\tau \delta$. While $\cL(x^*, \lambda^*)$ has no penalty
  since all constraints are satisfied, $\cL(x, \lambda')$ has penalty $2 \tau
  \delta$ since there is $\delta$ total constraint violation.  Thus,
  \[
    \cL(x, \lambda') \leq \cL(x^*, \lambda^*) - 2 \tau\delta + \tau \delta <
    \OPT.
  \]
  Thus, any infeasible $x$ gets payoff at most $\OPT$ in the worst case over
  dual strategies $\Lambda$.  Since $x^*$ is a primal play achieving payoff
  $\OPT$, the value of the game must be exactly $\OPT$. In particular, $x^*$ is
  a maxmin strategy.

  Now, since the primal player and the dual player play in compact sets $S,
  \Lambda$, the minmax theorem \citep{Kneser52} states that the restricted game
  has an equilibrium $(x^\bullet, \lambda^\bullet)$. Thus, $\lambda^\bullet \in
  \Lambda$ is a minmax strategy, and $(x^*, \lambda^\bullet)$ is the claimed
  equilibrium.
\end{proof}

In the remainder, we will always work with the restricted game: the dual player
will have action set $\Lambda = \{ \lambda \in \RR^k_+ \mid \| \lambda \|_2 \leq
2 \tau \sqrt{k} \}$.
\fi
To show that \PDD computes an approximate equilibrium, we
want to use the no-regret guarantee \preflong{thm:noisy-GD}. We define the
second key parameter for accuracy.

\begin{definition}
  Consider all problem instances $\cO = (S, v, c, b)$ for a class of convex
  optimization problems. The class has \emph{width bounded by $w$} if
  \begin{smalldisplay}
    \max \left| \sum_{i = 0}^n c^{(i)}_j (x^{(i)}) - b_j \right| \leq w ,
  \end{smalldisplay}
  where the max is taken over all instances $\cO$, and coupling constraint $j$,
  and $x \in S$.
\end{definition}
\ifshort\else
\begin{example*}
  The coupling constraints are of the form
  \[
    \sum_{i = 1}^n x^{(i)}_{tq} \leq c_{tq} ,
  \]
  and the $x$ variables lie in $[0, d_{max}]$. Let $c_{max} = \max_{t, q}
  c_{tq}$ be the maximum capacity over all slots. If we assume there are a large
  number of agents so $n d_{max} \gg c_{max}$, then the width is bounded by  $w
  = n d_{max}$.
\end{example*}
\fi

The width controls how fast online gradient descent converges. However, although
the convergence time depends polynomially on $w$, our accuracy bound depends
only on $\log(w)$.
\ifshort\else
Applying \Cref{thm:noisy-GD} gives the following regret guarantee for the dual
player.

\begin{lemma} \label{lem:regret}
  Suppose $w$ is the width for a class of convex optimization problems.  Then
  with probability at least $1 - \beta$, running \PDD for $T =
  w^2$ iterations yields a sequence of dual plays $\lambda^{(1)}, \dots,
  \lambda^{(T)} \in \Lambda$ with regret $\cR_p$ to any point in $\Lambda$,
  against the sequence of best responses $x_1, \dots, x_T$, where
  \[
    \cR_p =
    O \left( \frac{ k \tau \sigma \log^{1/2}(w/\delta) }{\epsilon} \log
      \frac{wk}{\beta} \right) .
  \]
  As above, $\sigma$ is the gradient sensitivity and $\tau$ is the dual bound.
\end{lemma}
\begin{proof}
  We add Gaussian noise with variance
  \[
    \frac{2 \sigma^2 \log \left( 1.25 T/\delta \right)}{\eps'^2}
  \]
  to each gradient, where
  \[
    \epsilon' = \frac{\epsilon}{\sqrt{8 k T \log(1/\delta)}} .
  \]
  Furthermore, the target space $\Lambda$ has $\ell_2$ diameter $2\tau \sqrt{k}$.
  So by \Cref{thm:noisy-GD}, the regret to the sequence of unnoised best
  responses $x_1, \dots, x_T$ is at most
  \begin{align*}
    \cR_p &\leq \frac{\tau \sqrt{k}}{\sqrt{T}}\left( w +
      4 \sqrt{\frac{\pi \log(1.25 T/\delta)}{\log 2}}
        \frac{\sigma}{\epsilon'} \log\left(\frac{2Tk}{\beta} \right)
    \right) \\
    &\leq \tau \sqrt{k} \left( 1 +
        \frac{20 \sigma \sqrt{8 k}}{\epsilon}
        \log\left(\frac{2T k}{\beta} \right) \log^{1/2} \left( \frac{T}{\delta}
        \right) \right) \\
    &\leq \frac{40 \sqrt{8} k \tau \sigma }{\epsilon} \log\left(\frac{2 w^2
        k}{\beta} \right) \log^{1/2} \left( \frac{w^2}{\delta} \right) \\
    &= O \left( \frac{ k \tau \sigma  }{\epsilon} \log
      \frac{wk}{\beta} \log^{1/2} \frac{w}{\delta} \right)
  \end{align*}
  as desired.
\end{proof}

By applying \Cref{thm:equil-restr}, we immediately know that \PDD computes an
approximate equilibrium.

\begin{corollary} \label{cor:equil}
  With probability at least $1 - \beta$, the output $(\bar{x}, \bar{\lambda})$
  of \PDD is an $\cR_p$-approximate equilibrium for
  \[
    \cR_p =
    O \left( \frac{ k \tau \sigma \log^{1/2}(w/\delta) }{\epsilon} \log
      \frac{wk}{\beta} \right) .
  \]
\end{corollary}
\fi

Finally, we show that \PDD generates an approximately feasible and
optimal point. For feasibility, the intuition is simple: since the dual player
can decrease the Lagrangian value by putting weight $2 \tau$ on every violated constraint,
there can't be too many violated constraints: $\bar{x}$ is an approximate
equilibrium strategy for the primal player, and hence achieves nearly $\OPT$ even against a best response from the dual player.

Proving approximate optimality is similar. We think of the primal player as
deriving payoff in two ways: from the original objective, and from any
over-satisfied constraints with positive dual variable. The dual player can
always set the dual variables for over-satisfied constraints to zero and decrease
the Lagrangian value. If the primal player derives large payoff from these
over-satisfied constraints, then the best response by the dual player will
substantially reduce the payoff of the primal player and lead to a
contradiction with the approximate maxmin condition. More formally, we have the
following theorem.

\begin{theorem} \label{thm:bigthm}
  Suppose a convex program $\cO = (S, v, c, b)$ has width bounded by $w$ and
  gradient sensitivity bounded by  $\sigma$, and dual bound $\tau$.  With
  probability at least $1 - \beta$, \PDD produces a point $\bar{x} \in S$ such
  that:
  \begin{itemize}
    \item the total violation of coupling constraints is bounded by
      \[
        \sum_{j = 1}^k \sum_{i = 0}^n \left(
          c^{(i)}_j (\bar{x}^{(i)}) - b_j \right)_+ =
        O \left( \frac{ k \sigma \log^{1/2}(w/\delta) }{\epsilon} \log
          \frac{wk}{\beta} \right) , \qquad \text{and}
      \]
    \item
      the objective satisfies

      \[
        \sum_{i = 0}^n v^{(i)} (\bar{x}^{(i)}) \geq \OPT - \alpha ,
        \qquad \text{for} \qquad
        \alpha = 2 \cR_p
        = O \left( \frac{ k \tau \sigma \log^{1/2}(w/\delta) }{\epsilon} \log
          \frac{wk}{\beta} \right) .
      \]
  \end{itemize}
\end{theorem}
\ifshort\else
\begin{proof}
   By \Cref{lem:regret}, \PDD computes an $\cR_p$ approximate equilibrium
   $(\bar{x}, \bar{\lambda})$ with
  \[
    \cR_p = O \left( \frac{ k \tau \sigma \log^{1/2}(w/\delta) }{\epsilon} \log
      \frac{wk}{\beta} \right) ,
  \]
  with probability at least $1 - \beta$.

  Let us consider first consider feasibility. Since $S$ is convex and each best
  response lies in $S$, $\bar{x} \in S$. Suppose $\bar{x}$ violates the coupling
  constraints by
  \[
    \Delta_1 := \sum_{j = 1}^k \left( \sum_{i = 0}^n c^{(i)} (\bar{x}^{(i)}) -
      b_j \right)_+ .
  \]
  Define $\lambda' \in \Lambda$ as
  \[
    \lambda' =
    \begin{cases}
      2 \tau &\text{if } \sum_{i = 0}^n c^{(i)} (\bar{x}^{(i)}) > b_j \\
      0 &\text{otherwise.}
    \end{cases}
  \]
  Let's compare the payoff $\cL(\bar{x}, \bar{\lambda})$ to $\cL(\bar{x},
  \lambda')$. By the equilibrium property, we have
  \[
    \OPT - \cR_p \leq \cL(\bar{x}, \bar{\lambda}) \leq \OPT + \cR_p
  \]
  and
  \begin{equation} \label{eq:dual-dev}
    \cL(\bar{x}, \lambda') \geq \OPT - 2\cR_p .
  \end{equation}
  For $\cL(\bar{x}, \lambda')$, since the dual bound is $\tau$, we know
  \[
    \sum_{i = 0}^n v^{(i)} (\bar{x}^{(i)}) \leq \OPT + \tau \Delta_1 .
  \]
  At the same time, the penalty is at least
  \[
    \sum_{j = 1}^k \lambda'_j \left( \sum_{i = 0}^n c^{(i)}_j (\bar{x}^{(i)}) -
      b_j \right) \geq 2 \tau \Delta_1 .
  \]
  So,
  \[
    \cL(\bar{x}, \lambda') \leq \OPT - \tau \Delta_1
  \]
  and by \Cref{eq:dual-dev},
  \[
    \Delta_1 \leq \frac{2\cR_p}{\tau} =
    O \left( \frac{ k \sigma \log^{1/2}(w/\delta) }{\epsilon} \log
      \frac{wk}{\beta} \right)
  \]
  as desired.

  To show optimality, let the convex program have optimal objective $\OPT$.
  Suppose $\bar{x}$ has objective value:
  \[
    \sum_{i = 0}^n v^{(i)}(\bar{x}^{(i)}) = \OPT - \alpha .
  \]
  and say the penalty against $\bar{\lambda}$ is
  \[
    \rho = \sum_{j = 1}^k \bar{\lambda}_j \left( \sum_{i = 0}^n c^{(i)}_j
      (\bar{x}^{(i)}) - b_j \right)
  \]
  for total Lagrangian value $\cL(\bar{x}, \bar{\lambda}) = \OPT - \alpha - \rho$.
  Consider the deviation of the dual variables $\lambda' \in \Lambda$:
  \[
    \lambda'_j =
    \begin{cases}
      0 &\text{if coupling constraint } j \text{ loose} \\
      \bar{\lambda}_j &\text{otherwise.}
    \end{cases}
  \]
  Then $\cL(\bar{x}, \lambda') = \OPT - \alpha$. But since $(\bar{x},
  \bar{\lambda})$ is an $\cR_p$-approximate equilibrium,
  \[
    \cL(\bar{x}, \lambda') \geq \OPT - 2 \cR_p
  \]
  so $\alpha \leq 2 \cR_p$ as desired.
\end{proof}

\begin{remark*}
  We stress that \Cref{thm:bigthm} bounds the sum of the violations over all
  coupling constraints. In \Cref{sec:extensions}, we discuss how to modify the
  solution to instead give a stronger bound on the \emph{maximum} violation over
  any coupling constraint at a small cost to the objective value.
\end{remark*}

Applying \Cref{thm:bigthm} to the electricity scheduling LP,
we immediately have the following result.

\begin{corollary}
  With probability at least $1 - \beta$, \PDD run on the
  electricity scheduling LP produces an electricity schedule $\bar{x}$ that
  \begin{itemize}
    \item satisfies all demand constraints exactly,
    \item exceeds the power supply constraints by
      \[
        O \left( \frac{ QT \sqrt{d_{max} \log(nd_{max}/\delta) }}{\epsilon} \log
          \frac{nd_{max} TQ}{\beta} \right)
      \]
      in total, over all time slots, and
    \item achieves welfare at least $\OPT - \alpha$ for
      \[
        \alpha = O \left( \frac{ QT \sqrt{d_{max} \log(nd_{max}/\delta)
            }}{\epsilon} \log \frac{nd_{max} TQ}{\beta} \right) ,
      \]
      where $\OPT$ the optimal objective value.
  \end{itemize}
\end{corollary}
\fi

\ifshort
\section{Further Examples}
\else
\section{Examples}
\fi
 \label{sec:examples}

In this section, we illustrate the general bounds for \PDD by instantiating
them on several example problems. \ifshort For lack of space, we only present
one example, deferring the remaining examples to the extended version. We
summarize the bounds we get applying our main result to different examples in
the following table. The details of these and other examples are in the full
version. \else For each example, we present the problem description and the
relevant parameters (gradient sensitivity, dual bounds, and width), and then
state the guarantee we get on the quality of the solution produced by \PDD.
\fi

\begin{table}[h]
\begin{center}
\begin{tabular}{|l|l|l|l|l|}
\hline
Problems  & \begin{tabular}[c]{@{}l@{}}Relevant \\
Parameters\end{tabular} &
\begin{tabular}[c]{@{}l@{}} Welfare / Cost\\ ($\OPT
  \pm$)\end{tabular} &
\begin{tabular}[c]{@{}l@{}} Constraint\\Violation\end{tabular}\\
\hline
\begin{tabular}[c]{@{}l@{}}$d$-demand\\ Allocation\end{tabular}
& \begin{tabular}[c]{@{}l@{}}$d$: max  bundle size;\\$k$: \#
  goods\end{tabular}    & $\tilde O\left(kd/\eps\right)$ &$\tilde O\left(kd/\eps\right)$ \\ \hline
\begin{tabular}[c]{@{}l@{}}Multi-commodity\\ Flow\end{tabular} &  \begin{tabular}[c]{@{}l@{}}$L$: longest
  path;\\$m$: \# edges\end{tabular}       &  $\tilde O\left(m\sqrt{L}/\eps\right)$&  $\tilde O\left(m\sqrt{L}/\eps\right)$\\ \hline
\begin{tabular}[c]{@{}l@{}}Multi-dimensional\\ Knapsack\end{tabular}
&   \begin{tabular}[c]{@{}l@{}}$v_i$: value; $\{w_{ij}\}$: weights \\
  $k$: \# resources\end{tabular}    &  $\tilde O\left(k^{3/2}\max_{i,j}\frac{v_i}{w_{ij}}/\eps\right)$&  $\tilde O\left(k^{3/2}/\eps\right)$\\ \hline
\ifshort\else
\begin{tabular}[c]{@{}l@{}}Allocation with \\ shared
resources\end{tabular} & \begin{tabular}[c]{@{}l@{}}$m$: \# projects; $k$: \# resources \\
  $d$: \# resources a project needs\end{tabular}&   $\tilde O\left(md^{3/2}/\eps\right)$&   $\tilde O\left(md^{3/2}/\eps\right)$\\ \hline
\begin{tabular}[c]{@{}l@{}}Aggregative Games \\
   Equilibrium LP\end{tabular} &  \begin{tabular}[c]{@{}l@{}}$k$:
   dimension of aggregator;\\$\gamma$: sensitivity of
   aggregator \end{tabular} & N/A & $\tilde O\left(k^{3/2}\gamma/\eps \right)$\\ \hline
 \fi
\end{tabular}
\end{center}
\end{table}

\ifshort\else
While some examples are combinatorial optimization problems, our
instantiations of $\PDD$ only produce fractional solutions. To
simplify the presentation, we will not discuss rounding techniques
here. In~\Cref{sec:rounddude}, we present a extensions of $\PDD$ to
privately round the fractional solution with a small additional loss.
\fi
\ifshort \vspace{-1cm} \fi
\SUBSECTION{The $d$-Demand Allocation Problem}
Consider a market with $n$ agents, a collection of goods $G$ of $k$
different types, and let $s_j$ be the supply of good $j$. We assume
that the agents have \emph{$d$-demand valuations} over bundles of
goods, i.e., they demand bundles of size no more than
$d$. Let $\cB = \{S\subseteq G\mid |S|\leq d\}$ denote the set of all
bundles with size no more than $d$. For each $S\in \cB$ and $i\in[n]$,
we write $v^{(i)}(S)$ to denote agent $i$'s valuation on $S$; we
assume that $v^{(i)}(S) \in [0,1]$.
We are interested in computing a welfare-maximizing allocation:
\[
\begin{aligned}
  & \textrm{maximize} & & \sum_{i = 1}^n \sum_{S\in \cB} v^{(i)}(S) \cdot
  x^{(i)}(S) & & \\
  & \textrm{subject to} & & \sum_{i = 1}^n \sum_{S\ni j} x^{(i)}(S) \leq s_j & &
  (\textrm{for } j\in [k]) \\
  & & & \sum_{S\subseteq \cB} x^{(i)}(S) \leq 1, \quad x^{(i)}(S) \geq 0 &&
  (\textrm{for } i \in[n], S\in \cB) .
\end{aligned}
\]
The private data lies in agents' valuations over the bundles, and two
problem instances are neighboring if they differ by any agent $i$'s
valuations. Since each agent demands at most $d$ items, the gradient
sensitivity $\sigma$ is at most $\sqrt{2}d$,
and the width is bounded by $nd$. Also, the problem has a dual bound $\tau = 1$,
as each agent's valuation is bounded by 1.
\begin{corollary}\label{cor:ddemand}
  With probability at least $1-\beta$, $\PDD(\cdot,
  \sqrt{2}d, 1, nd, \eps, \delta, \beta)$ computes a fractional
  allocation $\bar x$ such that the total supply violation is bounded
  by
\[
\alpha = O\left( \frac{kd\,{\log^{1/2}(nd/\delta)}}{\eps}
  \log{\frac{ndk}{\beta}}\right)
\quad \text{with welfare} \quad
\sum_{i = 1}^n \sum_{S\subseteq G} v^{(i)}(S) \cdot \bar x^{(i)}(S)
\geq \OPT - \tau \alpha = \OPT - \alpha.
\]
\end{corollary}

Note that the \emph{average} violation per good is  $\tilde{O}(d/\epsilon)$.
See \thelongref{sec:extensions} for a
method that solves this problem with \emph{no} constraint violations,
at a small cost to the objective.

\begin{remark*}
  Our result gives an affirmative answer to the open problem posed
  by~\citet{HHRRW14}: can the allocation problem be solved privately for more
  general valuation functions beyond the class of gross substitutes (GS)?  Our
  welfare and supply violation bounds are incomparable with the ones in their
  work, which assume GS valuations.  For a detailed discussion,
  see~\thelongref{sec:extensions}.

  Note that there are exponentially many primal variables (in $d$),
  but our error bound is independent of the number of variables.
  Moreover, we can
  implement $\PDD$ efficiently given a \emph{demand oracle} for each
  player, which, for any given item prices $\{\lambda_j\}$,
  returns a bundle $B \in \arg\max_{S\in \cB} (v^{(i)}(S) -
  \sum_{j\in S} \lambda_j )$ for each agent $i$.
\end{remark*}
\ifshort\else

\SUBSECTION{Multi-Commodity Flow}
While a broad class of combinatorial problems are
instantiations of the $d$-demand allocation problem, some problems have
special structure and more compact representations. For example,
consider the following multi-commodity flow problem over a network
$G(V, E)$. There are $m$ edges and $n$ agents, and each agent needs to
route $1$ unit of flow from its source $s_i$ to its destination $t_i$.
We assume that for any agent $i$, any path from $s_i$ to $t_i$ has
length bounded by $L$. For each edge $e\in E$, there is an associated
cost $c^{(i)}_e$ if an agent $i$ uses that edge, and we assume that
$c_e^{(i)}\in [0, 1]$ for all $e$ and $i$. Also, each edge $e$ has a
\emph{capacity constraint}: the amount of flow on edge $e$ should be no
more than $q_e$. The problem can be written as the following LP:
\[
\begin{aligned}
  & \textrm{minimize} & &
    \sum_{i=1}^n\sum_{e\in E} c_e^{(i)} \cdot x^{(i)}_e & & \\
  & \textrm{subject to} & & \sum_{i = 1}^n x^{(i)}_e \leq q_e & & (\textrm{for
    each } e\in E) \\
  & & & x^{(i)} \mbox{ forms a }(s_i, t_i)\mbox{-flow} & & (\textrm{for each }
  i\in [n]) .
\end{aligned}
\]
The private data lies in each agent's costs on the edges and its
source and destination. Since each agent only routes $1$ unit of flow,
the gradient sensitivity of the problem is bounded by $\sigma =
\sqrt{2L}$.
The problem has a dual bound $\tau = 1$, and its width is bounded by
$n$.
\begin{corollary}
  With probability at least $1-\beta$, $\PDD(\cdot,
  \sqrt{2L}, 1, n, \eps, \delta, \beta)$ computes a fractional flow
  $\bar x$ such that the total capacity violation is bounded by
\[
  \alpha = O\left(\frac{m \sqrt{L}
      \log^{1/2}(n/\delta)}{\eps}\log{\frac{nm}{\beta}} \right),
\]
and the resulting fractional flow has total cost at
most $\OPT + \alpha\tau = \OPT + \alpha$.
\end{corollary}
Note again that this is the \emph{total} violation summed over all edges. The
\emph{average} violation per edge is smaller by a factor of $m$:
$\tilde{O}(\sqrt{L}/\epsilon)$. See \Cref{sec:extensions} for a method
that solves this problem with \emph{no} constraint violations, at a small cost
to the objective.

\SUBSECTION{Multi-Dimensional Knapsack}
In a multi-dimensional knapsack problem, there are a set of $n$ items
with values $v_i \in [0,1]$ and $k$ different resources with
capacities $c_i >0$. Each item $i$ requires an amount
$w_{ij}\in[0,1]$ of each resource $j$. The goal is to select a
subset of items to maximize the sum value while satisfying the resource
constraints:
\[
\begin{aligned}
  & \textrm{maximize} & & \sum_{i=1}^n\sum_{e\in E} v_i\cdot x^{(i)} & & \\
  & \textrm{subject to} & & \sum_{i = 1}^n w_{ij} \, x^{(i)}  \leq c_j & &
  (\textrm{for each } j\in [k]) \\
  & & & x^{(i)}\in [0, 1]  & & (\textrm{for each }i\in [n]) .
\end{aligned}
\]

Each agent's private data consists of both the value of her job $v_i$
and its resource demands $\{w_{ij}\}$. The gradient sensitivity is
bounded by $\sigma = \sqrt{k}$, because each item can consume at most
1 unit of each resource. The problem has a dual bound $\tau = \max_{i,
  j}\frac{v_i}{w_{ij}}$,
and the width is bounded by $n$.
\begin{corollary}
  With probability at least $1-\beta$,
  $\PDD(\cdot, \sqrt{k}, 1, n, \eps, \delta, \beta)$ computes a
  fractional assignment such that the total violation in the resource
  constraints is bounded by
  \[
    \alpha =
    O\left(\frac{k\sqrt{k}\log^{1/2}(n/\delta)}{\eps}\log{\frac{nk}{\beta}}\right),
  \]
  and has total  profit at least
  \[
  \OPT - \tau\alpha = \OPT -\alpha \cdot
  \max_{i, j} \frac{v_i}{w_{ij}}.
\]
\end{corollary}
\SUBSECTION{Allocations with Shared Resources}
We now give an example with auxiliary decision variables that
are not associated with private data. Suppose we have $n$ agents,
$m$ projects, and $k$ different resources. Each agent $i$ has private
valuations $\{v_{ij}\}$ over the projects. Each project requires a set
of resources in $R_j$, but the resources can be shared between
different projects. A unit of resource $r$ has cost $c_r$, and for any
project $j$ with $e_j$ enrolled agents, we require at least $e_j$
units of resources $r$ for every $r\in R_j$. We also assume that each
project requires at most $d$ distinct resources, and so the number of
coupling constraints is bounded by $md$. We further assume that
$v_{ij}, c_r\in [0, 1]$ for all $i, j$ and $r$. Our goal is to match
people to projects so that the welfare of the agents minus the total
cost of the resources is maximized, as illustrated by the following linear
program:
\[
\begin{aligned}
  & \textrm{maximize} & & \sum_{i = 1}^n \sum_{j=1}^m v_{ij}\, x_{ij} -
  \sum_{r=1}^k c_r \, y_r & & \\
  & \textrm{subject to} & & \sum_{i = 1}^n x_{ij} \leq y_r
  & & (\text{for } j\in [k]\text{ and } r\in R_j) \\
  & & & \sum_{j=1}^k x_{ij} \leq 1, \quad x_{ij}\geq 0 & & (\text{for } i
  \in[n], j\in[k]) .
\end{aligned}
\]
The private data lies in the valuations of the agents over the
projects, and two problem instances are neighboring if they only
differ in some agent $i$'s preferences over the projects. To fit this
problem into our general framework, we interpret the variables
$\{y_r\}$ as the ``public'' variables, controlled by ``agent $0$''. The
gradient sensitivity is bounded by $\sigma = \sqrt{2d}$. The problem
has a dual bound $\tau = 1$, and the width of the problem is bounded
by $w = n$.
\begin{corollary}
  With probability at least $1-\beta$, $\PDD(\cdot,
  \sqrt{2d}, 1, n, \eps,\delta, \beta)$ computes a fractional
  allocation $\bar x$ such that the total resource violation (resource
  shortage across all projects) is bounded by
\[
\alpha =
O\left(\frac{md\sqrt{d}\log^{1/2}(n/\delta)}{\eps}\log{\frac{nmd}{\beta}} \right),
\]
and the project matching along with the resource allocation gives
welfare at least $\OPT - \alpha\tau = \OPT - \alpha$.
\end{corollary}
\SUBSECTION{Equilibrium Computation in Aggregative Games}
A recent paper by~\citet{CKRW14} showed
that mixed strategy equilibria in \emph{aggregative games}\footnote{%
  They consider multi-dimensional aggregative games, a broad class of games that
  generalizes both anonymous games and weighted congestion games. For more
  details, see~\cite{CKRW14}.}
can be computed using an algorithm that repeatedly solves a
feasibility linear programs of the following form:
\[
\begin{aligned}
  \sum_{i=1}^n \sum_{\ell=1}^m c_{i\ell}^j\cdot x_{i\ell} &\leq \hat s
  & & (\text{for all } j\in [k]) \\
  -\sum_{i=1}^n \sum_{\ell=1}^m c_{i\ell}^j\cdot x_{i\ell} &\leq -\hat s
  & & (\text{for all } j\in [k]) \\
  x^{(i)} &\in B_i & & (\text{for all } i\in[n]) .
\end{aligned}
\]

Each agent $i$ controls the variables $x^{(i)} = (x_{i1}, \ldots ,
x_{im})$, which forms a probability distribution over $m$ actions. We
assume that each coefficient in the coupling constraint is bounded by
$|c_{i\ell}^j|\leq \gamma$. Note that we can add the objective $\min_x
0$ the LP without changing the problem, so it can be framed as a
linearly separable convex program. In particular, the gradient sensitivity is
bounded by $\sqrt{k} \gamma$. Since the objective function is a
constant, any positive number is a dual bound for this problem,
so we could use $\tau = 1$ as a dual bound. Also, the width is bounded
by $\gamma n$.


\begin{corollary}
  With probability at least $1-\beta$, $\PDD(\cdot,
  \sqrt{k}\gamma, 1, \gamma n, \eps, \delta, \beta)$ outputs a mixed
  strategy profile that has total violation across all of the
  constraints bounded by
\[
\alpha =
O\left(\frac{k\sqrt{k}\gamma\log^{1/2}(1/\delta)}{\eps}\log{\frac{nk\gamma}{\beta}}
\right).
\]
\end{corollary}
\begin{remark*}
  This leads to an improvement over the private LP
  solver in~\citet{CKRW14}, which gives a violation bound of
  $\tilde{O}\left(\frac{\sqrt{n}\gamma}{\sqrt{\eps}} \right)$. Since
  $n$ is large and $k$ is a constant in their setting, our violation
  bound is considerably better with no the polynomial dependence
  on $n$.
\end{remark*}
\fi

\section{Variations on a Theme} \label{sec:extensions}

\ifshort
\else
When the convex program has additional structure, we can extend \PDD to achieve
additional guarantees.  In this section, we'll discuss two extensions to \PDD:
deriving prices to achieve approximate truthfulness, and guaranteeing exact
feasibility.
\fi

\SUBSECTION{Achieving Approximate Truthfulness}

In the course of computing an approximate (primal) solution to the convex
program, $\PDD$ also computes an approximate dual solution, which has a
standard interpretation as \emph{prices} (e.g, see \citet{cvxbook}).
Informally, if we think of each constraint as modeling a finite resource that
is divided between the variables, the dual variable for the constraint
corresponds to how much each variable should ``pay'' for using that resource.
If each agent has a real-valued \emph{valuation function} for their portion
of the solution, and the objective of the convex program is the sum of the
valuation functions (a \emph{social welfare} objective), then we can make the
prices interpretation precise: Each agent's solution will approximately
maximize her valuation less the prices charged for using each constraint,
where the prices are the approximate dual solution produced by \PDD.

Since the dual solution (and hence the final price vector) is computed under
standard differential privacy, we can also guarantee \emph{approximate
truthfulness}: an agent can't substantially increase her expected utility by
misreporting her private data. Informally, this is because agents can only
influence the prices to a small degree, so since the algorithm is maximizing
their utility function subject to the final prices (which they have little
influence on), agents are incentivized to report their true utility. One
technical difficulty is that since we are computing only an approximate
primal and dual solution, there may be a small number of agents who are not
getting their approximately utility maximizing allocation. For approximate
truthfulness, we will modify their allocations to assign them to their
favorite solution at the dual prices. This may further violate some primal
constraints, but only by a small amount.
\ifshort
For lack of space, we defer further details of our algorithm, \TPDD, to the
extended version.
\else

Let us first define the class of optimization problems we will solve truthfully.

\begin{definition} \label{def:welfare-cvx}
  Let the data universe be $\cU$, and suppose there are $n$ individuals.  A
  class of \emph{welfare maximization problems} is a class of convex programs
  $\cO = (S, v, c, b)$ associated to $\cU$, with the following additional
  properties:
  \begin{itemize}
    \item Bounded welfare: $v^{(i)} (x^{(i)}) \leq V$ for all agents $i$ and
      points $x^{(i)} \in S$, for all feasible sets $S$ for agent $i$ in the
      class.
    \item Bounded constraints: $\sum_{j = 1}^k c^{(i)}_j (x^{(i)}) \leq C_1$ for
      all agents $i$ points $x^{(i)} \in S$, for all feasible sets $S$ for agent
      $i$ in the class.
    \item Null action: for each agent $i$, there exists $x^{(i)} \in S^{(i)}$
      such that $v^{(i)} (x^{(i)}) = c^{(i)}_j (x^{(i)}) = 0$ for all
      constraints $j$. We call such a point a \emph{null action for $i$}.
  \end{itemize}
\end{definition}

Now, we can define the personal data and utility function for each player. As is
typical in the literature, agents' utilities will be quasilinear in money.

\begin{definition} \label{def:utility}
  Recall that agent $i$ has a compact \emph{feasible set} $S^{(i)} \subseteq
  \RR^d$, \emph{valuation function} $v^{(i)} : \RR^d \rightarrow \RR$, and
\emph{constraint functions} $c^{(i)}_j : \RR^d \rightarrow \RR$.  We assume that
$v^{(i)}(x^{(i)}) = 0$ for any $x^{(i)} \notin S^{(i)}$. An agent's
\emph{utility} for solution $x^{(i)}$ and price $p^{(i)}(x^{(i)}) \in \RR$ is
$v^{(i)}
(x^{(i)}) - p^{(i)}(x^{(i)})$.
\end{definition}

Our truthful modification to \PDD will assign each constraint a
price $\lambda_j \in \RR$, and charge each agent $i$ price
\[
  p^{(i)}(x^{(i)}) = \sum_{j = 1}^k \lambda_j c^{(i)}_j (x^{(i)}) ,
\]
where agent $i$'s solution is $x^{(i)}$. To guarantee truthfulness, we want
every agent to have a solution that is approximately maximizing her utility
given the fixed prices on constraints. This motivates the following definition:

\begin{definition} \label{def:satisfied}
  Let $\alpha \geq 0$ and prices $\lambda \in \RR^k$ be given. An agent $i$ with
  solution $x^{(i)}$ is \emph{$\alpha$-satisfied} with respect to these prices
  if
  \[
    v^{(i)} (x^{(i)}) - \sum_{j = 1}^k \lambda_j c^{(i)}_j (x^{(i)})
    \geq
    \max_{x \in S^{(i)}} v^{(i)} (x) - \sum_{j = 1}^k \lambda_j c^{(i)}_j (x) -
    \alpha .
  \]
\end{definition}

We are now ready to present our approximately truthful mechanism for
welfare maximization. The idea is to run \PDD, obtaining solution
$\bar{x}$ and approximate dual solution $\bar{\lambda}$, which we take
to be the prices on constraints. For $\alpha$ to be specified later,
we change the allocation for each agent $i$ who is not
$\alpha$-satisfied to a primal allocation $\tilde{x}^{(i)}$ that
maximizes her utility at prices $\bar{\lambda}$. Combining $\tilde{x}$
with $\bar{x}$ for $\alpha$-satisfied agents gives the final solution.
We call this algorithm $\TPDD$, formally described in \Cref{alg:truthful}.

\begin{algorithm}[h]
  \caption{$\TPDD(\cO, \sigma, \tau, w, \eps, \delta, \beta)$}
  \label{alg:truthful}
  \begin{algorithmic}
    \STATE{{\bf Input}: Welfare maximization problem $\cO = (S, v, c,
      b)$ with $n$ agents and $k$ coupling constraints, gradient
      sensitivity bounded by $\sigma$, a dual bound $\tau$, width
      bounded by $w$, and privacy parameters $\eps > 0, \delta \in (0,
      1/2)$, confidence parameter $\beta \in (0, 1)$, and truthfulness
      parameter $\alpha$.}

    \STATE{{\bf Run} $\PDD$:
      \[
        (\bar{x}, \bar{\lambda}) := \PDD(\cO, \sigma, \tau, w, \epsilon, \delta,
        \beta) .
      \]
    }
    \STATE{{\bf for each } agent $i = 0 \dots n$:}
    \INDSTATE[1]{Let the price of the solution be:
      \[
        p^{(i)} (x) := \sum_{j = 1}^k \bar{\lambda}_j c^{(i)}_j (x) .
      \]
    }
    \INDSTATE[1]{{\bf if } $\bar{x}^{(i)}$ does not satisfy
      \[
        v^{(i)} (\bar{x}^{(i)}) - p^{(i)}(x^{(i)})
        \geq \max_{x \in S^{(i)}} v^{(i)} (x) - p^{(i)}(x) - \alpha ,
      \]
    }
    \INDSTATE[2]{Set
      \[
        \bar{x}^{(i)} := \argmax_{x \in S^{(i)}} v^{(i)} (x) - p^{(i)}(x) .
      \]
    }
    \STATE{{\bf Output}: $\bar{x}^{(i)} \text{ and price } p^{(i)}(x^{(i)})
      \text{ to agents } i \in [n]$.
    }
  \end{algorithmic}
\end{algorithm}

Let's first show that the final allocation is approximately feasible, and
approximately maximizing the objective (the \emph{social welfare}). Both proofs
follow by bounding the number of $\alpha$-unsatisfied agents, and arguing that
since the intermediate solution $(\bar{x}, \bar{\lambda})$ is an approximate
equilibrium (by \Cref{cor:equil}), changing the allocation of the unsatisfied agents will only degrade the feasibility and optimality a bit more (beyond what is
guaranteed by \Cref{thm:bigthm}).

\begin{lemma} \label{lem:unsat}
  Let $\alpha > 0$ be given, and let $(\bar{x}, \bar{\lambda})$ be an
  $\cR_p$-approximate equilibrium of the Lagrangian game. Then, the number of
  bidders who are not $\alpha$-satisfied is at most $\cR_p/\alpha$.
\end{lemma}
\begin{proof}
  Note that the Lagrangian can be written as the sum of agent utilities plus a
  constant:
  \[
    \cL(x, \lambda) = \sum_{i = 0}^n \left( v^{(i)} (x^{(i)}) -
    \sum_{j = 1}^k \lambda_j c^{(i)}_j(x^{(i)}) \right) + \sum_{j = 1}^k
    \lambda_j b_j
    = \sum_{i = 0}^n u^{(i)} (x^{(i)}, p^{(i)}) + \sum_{j = 1}^k \lambda_j b_j ,
  \]
  so every $\alpha$-unsatisfied agent that deviates to her favorite solution at
  prices $\lambda$ increases the Lagrangian by at least $\alpha$; suppose that
  there are $m$ such bidders.  Since $(\bar{x}, \bar{\lambda})$ is an
  $\cR_p$-approximate equilibrium, we can bound the change in the Lagrangian by
  \[
    \alpha m \leq \cL(x, \bar{\lambda}) -  \cL(\bar{x}, \bar{\lambda}) \leq
    \cR_p .
  \]
  So, at most $m \leq \cR_p/\alpha$ agents can be $\alpha$-unsatisfied.
\end{proof}

This immediately gives us approximate feasibility and optimality for the output
of \TPDD:

\begin{corollary} \label{cor:truthful-acc} Let $\beta > 0$. Running
  \TPDD on a welfare optimization problem $\cO = (S, v, c, b)$ with
  gradient sensitivity bounded by $\sigma$, a dual bound $\tau$, and
  width bounded by $w$, produces a point $\bar{x}$ such that all
  agents are $\alpha$-satisfied, and with probability at least $1 -
  \beta$:
  \begin{itemize}
    \item the maximum violation of any constraint is at most $\cR_p(2/\tau +
      C_1/\alpha)$; and
    \item the welfare $\sum_{i = 0}^n v^{(i)} (\bar{x}^{(i)})$ is at least $\OPT
      - \cR_p(2 + V/\alpha)$, where $\OPT$ is the optimal welfare.
  \end{itemize}
  As above,
  \[
    \cR_p = O \left( \frac{ k \tau \sigma \log^{1/2}(w/\delta) }{\eps} \log
      \frac{wk}{\beta} \right) .
  \]
\end{corollary}
\begin{proof}
  Follows directly from \Cref{thm:bigthm}, since any unsatisfied agent changes
  the welfare by at most $V$ and has total contribution to all constraints at
  most $C_1$, and by \Cref{lem:unsat} there are at most $\cR_p/\alpha$ unsatisfied
  agents.
\end{proof}

Additionally, \TPDD is approximately \emph{individually rational}:
no agent will have large negative utility. More precisely:

\begin{lemma} \label{lem:IR}
  Suppose \TPDD compute solution $x^{(i)}$ and payment
  $p^{(i)}(x^{(i)})$ for agent $i$. Then,
  \[
    v^{(i)}(x^{(i)}) - p^{(i)}(x^{(i)}) \geq - \alpha.
  \]
\end{lemma}
\begin{proof}
  The null action in agent $i$'s feasible region has utility equal to $0$: the
  valuation is $0$, and the constraint functions are all $0$, so the payment is
  $0$. The claim follows, since
  \[
    v^{(i)}(x^{(i)}) - p^{(i)}(x^{(i)}) \geq \max_{x \in S^{(i)}} v^{(i)}(x) -
    p^{(i)}(x) - \alpha .
  \]
\end{proof}

Now, let us turn to showing approximate truthfulness. We now suppose that agents
may misreport their feasible set $S^{(i)}$ and valuation function $v^{(i)}$, but
not their constraint functions $c^{(i)}_j$; we want to show that agents can't
gain much in expected utility by misreporting their inputs.\footnote{%
  More precisely, we can allow agents to misreport their constraint functions as
  long as the functions are \emph{verifiable}: they must feel their true
  contribution to the constraint (e.g., in terms of payments) when computing
  their utility.}

More formally, we want to show \emph{approximate truthfulness}. We will give a
combined multiplicative and additive guarantee:

\begin{definition} \label{def:truthful}
  Fix a class of welfare maximization problems. Consider a randomized function
  $f$ that takes in a convex optimization problem $\cO = (S, v, c, b)$, and
  outputs a point $x^{(i)} \in S^{(i)}$ and a price $p^{(i)}(x^{(i)}) \in \RR$
  to each
  agent $i$.  We say $f$ is \emph{$(\rho, \gamma)$-approximately truthful} if every
  agent $i$ with true feasible set $S^{(i)}$ and valuation function $v^{(i)}$
  has expected utility satisfying
  \[
    \Ex{ f(\cO') }{ v^{(i)}(x^{(i)}) -  p^{(i)}(x^{(i)}) } \leq
    \rho \Ex{ f(\cO) } { v^{(i)}(x^{(i)}) - p^{(i)}(x^{(i)}) } + \gamma ,
  \]
  where $\cO = (S^{(i)}, S^{(-i)}, v^{(i)}, v^{(-i)}, c, b)$ is the problem when
  agent $i$ truthfully reports, and $\cO' = (S'^{(i)}, S^{(-i)}, v'^{(i)},
  v^{(-i)}, c, b)$ is the problem when agent $i$ deviates to any report
  $S'^{(i)}, v'^{(i)}$ in the class.

\end{definition}
\begin{remark*}
  We note that the inequality must hold for any data from other agents
  $S^{(-i)}, v^{(-i)}$, not necessarily their true data. Put another way,
  truthful reporting should be an approximate dominant strategy (in
  expectation), not just an approximate Nash equilibrium.
\end{remark*}

We use a standard argument for achieving truthfulness via differential privacy.
First, we show that conditioning on the prices being fixed, an agent can't gain
much utility by changing her inputs; this holds deterministically, since agents
are getting their approximately utility maximizing solution given the prices and
their reported utility.

\begin{lemma} \label{lem:truthful-br}
  Condition on \TPDD computing a fixed sequence of dual variables
  prices $\lambda^{(1)}, \dots, \lambda^{(T)}$, and let $\bar{\lambda} = (1/T)
  \sum_{t = 1}^T \lambda^{(t)}$ be the final prices.
  Suppose an agent $i$ has true feasible set $S^{(i)}$ and valuation $v^{(i)}$.
  Let $u^{(i)}(S, v, \{ \lambda^{(t)} \})$ be $i$'s utility (computed according
  to her true feasible set, valuation, and final prices $\bar{\lambda}$) for the
  output $\bar{x}^{(i)}$ \TPDD computes when the sequence of dual
  variables is $\{ \lambda^{(t)} \}$ and  $i$ reports feasible set $S$ and
  valuation $v$. Then,
  \[
    u^{(i)} (S', v', \{ \lambda^{(t)} \})
    \leq
    u^{(i)} (S^{(i)}, v^{(i)}, \{ \lambda^{(t)} \}) + \alpha
  \]
  for every set $S'$ and valuation $v'$.
\end{lemma}
\begin{proof}
  By \Cref{cor:truthful-acc}, we know all agents are $\alpha$-satisfied, so
  \[
    u^{(i)}(S^{(i)}, v^{(i)}, \{ \lambda^{(t)} \}) \geq \max_{x \in S^{(i)}}
    v^{(i)} (x) - \sum_{j = 1}^k \bar{\lambda}_j c^{(i)}_j(x) - \alpha .
  \]
  Now, we claim that
  \[
    \max_{x \in S^{(i)}} v^{(i)} (x) - \sum_{j = 1}^k \bar{\lambda}_j
    c^{(i)}_j(x)
    \geq
    v^{(i)} (x') - \sum_{j = 1}^k \bar{\lambda}_j c^{(i)}_j(x')
  \]
  for any $x'$. This is clear for $x' \in S^{(i)}$. For $x' \notin S^{(i)}$,
  we know $v^{(i)}(x') = 0$ so the right hand side is at most $0$. Since there
  is a null action in $S^{(i)}$, the left hand side is at least $0$, so the
  inequality is true for $x' \notin S^{(i)}$.

  In particular, letting $x'$ be $i$'s solution when reporting $(S', v')$
  against dual variables $\{ \lambda^{(t)} \}$, the right hand side is precisely
  $u^{(i)} (S', v', \{ \lambda^{(t)} \})$, and we have
  \[
    u^{(i)}(S^{(i)}, v^{(i)}, \{ \lambda^{(t)} \}) \geq u^{(i)}(S', v', \{
    \lambda^{(t)} \}) - \alpha
  \]
  as desired.
\end{proof}

Now, we use differential privacy. Since the sequence of dual variables (and
final dual prices) are computed under standard differential privacy, any agent
misreporting her input only has a limited effect on the prices. More formally,
we will use the following standard lemma about the expected value of a
differentially private mechanism.

\begin{lemma} [\citet{MT07}] \label{lem:ex-priv}
  Suppose we have a non-negative real valued mechanism $\cM : \cD \rightarrow
  \RR^+$ that is $(\epsilon, \delta)$-differentially private, and suppose that
  the output is bounded by $\cM(D) \leq B$ for all inputs $D \in \cD$. Then, if
  $D, D' \in \cD$ are neighboring inputs,
  \[
    \EE \left[ \cM(D) \right] \leq e^\epsilon \EE \left[ \cM(D') \right] +
    \delta B .
  \]
\end{lemma}

This is enough to argue that \TPDD is approximately truthful in
expectation.

\begin{theorem} \label{thm:truthful}
  Suppose we run \TPDD on a class of welfare optimization problems
  with a dual bound $\tau$. Then, \TPDD is $(\rho,
  \gamma)$-approximately truthful for
  \[
    \rho = e^\eps ,
    \qquad
    \gamma = \alpha(2e^\eps - 1) + \delta \max \{ V, C_1 \tau \sqrt{k} \} .
  \]
\end{theorem}
\begin{proof}
  Fix valuations and feasible sets of $n-1$ agents, the constraint functions of
  all agents, and consider a possibly deviating agent $i$. We first note that
  since the prices $\bar{\lambda}$ have norm $\| \bar{\lambda} \|_2 \leq \tau
  \sqrt{k}$, the maximum price charged is
    \[
      \< c^{(i)}(x), \bar{\lambda} \> \leq \|c^{(i)}(x)\|_2 \|\bar{\lambda}\|_2
      \leq C_1 \tau \sqrt{k}
  \]
  by Cauchy-Schwarz. Since the valuation is at most $V$, the utility of agent
  $i$ is bounded by $\max \{V, C_1 \tau \sqrt{k} \}$. As above, let $u^{(i)}(S, v,
  \{ \lambda^{(t)} \})$ be the true utility of agent $i$ for the outcome
  produced by \TPDD when reporting feasible set $S$ and valuation
  $v$, against dual variables $\lambda^{(t)}$. By approximate individual
  rationality (\Cref{lem:IR}), $u^{(i)}(S, v, \{ \lambda^{(t)} \}) + \alpha \geq
  0$. By \Cref{lem:truthful-br},
  \[
    u^{(i)}(S, v, \{ \lambda^{(t)} \} ) \leq u^{(i)}(S^{(i)}, v^{(i)}, \{
    \lambda^{(t)} \} ) + \alpha
  \]
  for every sequence of dual variables $\{ \lambda^{(t)} \}$. Let $g(\cO)$ be
  the sequence of dual variables produced by \TPDD on problem
  $\cO$. Noting that $u^{(i)}(S, v, -)$ is a deterministic function of the
  sequence of dual variables, which are differentially private,
  \Cref{lem:ex-priv} shows that
  \begin{align*}
    \Ex{\{ \lambda^{(t)} \} \sim g(\cO')}
    {u^{(i)} (S'^{(i)}, v'^{(i)}, \{ \lambda^{(t)} \})}
    &\leq \Ex{\{ \lambda^{(t)} \} \sim g(\cO)}
    {u^{(i)} (S^{(i)}, v^{(i)}, \{ \lambda^{(t)} \}) + \alpha}\\
    &\leq e^\eps
    \Ex{\{ \lambda^{(t)} \} \sim g(\cO)}
    {u^{(i)} (S^{(i)}, v^{(i)}, \{ \lambda^{(t)} \}) + 2\alpha}+ \delta \max
    \{V, C_1 \tau \sqrt{k} \} \\
    &\leq e^\eps
    \Ex{\{ \lambda^{(t)} \} \sim g(\cO)}
    {u^{(i)} (S^{(i)}, v^{(i)}, \{ \lambda^{(t)} \})} + \alpha (2e^\eps - 1) +
    \delta \max \{V, C_1 \tau \sqrt{k} \} ,
  \end{align*}
  where $\cO = (S^{(i)}, S^{(-i)}, v^{(i)}, v^{(-i)}, c, b)$ is the problem when
  agent $i$ truthfully reports, and $\cO' = (S'^{(i)}, S^{(-i)}, v'^{(i)},
  v^{(-i)}, c, b)$ is the problem when agent $i$ deviates to any report
  $S'^{(i)}, v'^{(i)}$ in the class.
\end{proof}

Finally, it is straightforward to show that \TPDD is private.

\begin{theorem} \label{true-priv}
  \TPDD satisfies $(\epsilon, \delta)$-joint differential privacy.
\end{theorem}
\begin{proof}
  By the privacy of \PDD (\Cref{thm:dual-priv,thm:primal-jdp}) and the billboard
  lemma (\Cref{billboard}).
\end{proof}
\fi

\SUBSECTION{Adding exact feasibility}
\label{sec:tightdude}

For the more restricted class of \emph{packing linear programs}, not only can
we achieve approximate truthfulness, but we can also achieve exact
feasibility. In fact, we can also round each agent's solution to a vertex in
their feasible region $S_i$ (which is integral for many problems); this can
lead to more natural solutions for many problems (e.g., matchings and flows).
\ifshort For lack of space, we defer further details to the extended version.
\else

Let's first define the linear programs we will consider.

\begin{definition} \label{def:packing}
  A class of \emph{packing linear programs} is a class of convex programs $\cO =
  (S, v, c, b)$ with the following additional properties:
  \begin{itemize}
    \item Objective functions are linear and bounded: $0 \leq v^{(i)}(x^{(i)})
      \leq V$ for all $x^{(i)} \in S^{(i)}$.
    \item Constraint functions are linear and bounded:
      $0 \leq c^{(i)}(x^{(i)}) \leq C_\infty$ for all $x^{(i)} \in S^{(i)}$.
    \item For each $i$, $0 \in S^{(i)}$.
  \end{itemize}
  The parameters $V, C_\infty$ should hold for the whole class (i.e., they are
  independent of the private data).  To emphasize that the objective and
  constraints are linear, we will often write $\< v^{(i)}, x^{(i)} \>$ for
  $v^{(i)}(x^{(i)})$, and likewise for $c^{(i)}_j$.
\end{definition}

\fi
Our algorithm \TRPDD will first \emph{tighten} the constraints, by reducing the
scalars $b_j$ by an amount $\xi$; we think of this step as ``reserving some
constraint''.  Next, we will run \PDD on the reduced problem. Like \TPDD, we
will let the final dual variables be the prices for the constraints, and we will
reassign agents who are very unsatisfied to their favorite good at the final
prices; this will guarantee approximate truthfulness. Finally, satisfied agents
will \emph{round}: they will each select one of their best responses $x^{(i)}_1,
\dots, x^{(i)}_T$ uniformly at random. Note that for packing linear problems,
each best response is a vertex of an agent's private feasible
region (if agents break ties by selecting vertices).
Since each agent's favorite good at the prices is also a vertex of the private
feasible region, all agents will end up playing at a vertex. We will choose the
amount $\xi$ to cover the potential increase in each constraint from the
unsatisfied agents and from the rounding, thereby giving exact feasibility.
\ifshort
For lack of space, we defer further details of our algorithm, \TRPDD, to the
extended version.
\else
We present the algorithm $\TRPDD$ in \Cref{alg:tight}.

\begin{algorithm}[h]
  \caption{$\TRPDD(\cO, \sigma, \tau, w, \eps, \delta, \beta)$}
  \label{alg:tight}
  \begin{algorithmic}
    \STATE{{\bf Input}: Packing linear program $\cO = (S, v, c, b)$ with
      $n$ agents and $k$ coupling constraints, gradient sensitivity
      bounded by $\sigma$, a dual bound $\tau$, width bounded by $w$, and
      privacy parameters $\eps > 0, \delta \in
      (0, 1/2)$, and confidence parameter $\beta \in (0, 1)$.}

    \STATE{{\bf Initialize}:
      \begin{mathpar}
        \xi :=
          \sqrt{3 b_j}
          +
          \frac{160 \sqrt{8} k \tau \sigma C_\infty}{\epsilon}
          \log^2 \left(\frac{4 w^2 k^2}{\beta} \right)
          \log^{1/2} \left( \frac{2w}{\delta} \right)
          \left( \frac{2}{\tau} + \frac{C_\infty k}{\alpha} \right)
        ,
        \and \cO_{red} := \cO \text{ with scalars } b_j - \xi .
      \end{mathpar}
    }

    \STATE{{\bf Run} $\PDD$:
      \[
        (\bar{x}, \bar{\lambda}) := \PDD(\cO_{red}, \sigma, \tau, w, \epsilon/2,
        \delta/2, \beta/2) .
      \]
    }
    \STATE{{\bf for each } agent $i = 0 \dots n$:}
    \INDSTATE[1]{Let the price of the solution be:
      \[
        p^{(i)} (x) := \sum_{j = 1}^k \bar{\lambda}_j c^{(i)}_j (x) .
      \]
    }
    \INDSTATE[1]{{\bf if } $\bar{x}^{(i)}$ does not satisfy
      \[
        v^{(i)} (\bar{x}^{(i)}) - p^{(i)}(x^{(i)})
        \geq \max_{x \in S^{(i)}} v^{(i)} (x) - p^{(i)}(x) - \alpha ,
      \]
    }
    \INDSTATE[2]{Set
      \[
        \tilde{x}^{(i)} := \argmax_{x \in S^{(i)}} v^{(i)} (x) - p^{(i)}(x) .
      \]
    }
    \INDSTATE[1]{{\bf else }}
    \INDSTATE[2]{Select $\tilde{x}^{(i)}$ uniformly at random from best responses
      $x^{(i)}_1, \dots, x^{(i)}_T$.
    }
    \STATE{{\bf Output}: $\tilde{x}^{(i)} \text{ and price } p^{(i)}(x^{(i)})
      \text{ to agents } i \in [n]$.}
  \end{algorithmic}
\end{algorithm}

To analyze the welfare and the constraint violation, let $\OPT_{red}$ be the
optimal objective of the reduced problem, and let $\kappa = \max_j \xi/b_j$ be
the largest fraction of constraint we are reducing. We assume that the problem
is feasible, so $\OPT_{red}$ is defined and $\kappa < 1$. We can immediately
lower bound $\OPT_{red}$.

\begin{lemma} \label{lem:opt-red}
  $\OPT_{red} \geq (1 - \kappa) \OPT$ .
\end{lemma}
\begin{proof}
  For any optimal solution $x^*$ of the original problem, $(1 - \kappa) x^*$ is
  a feasible solution of the reduced problem with objective $(1 - \kappa) \OPT$.
  The personal feasibility constraints aren't violated since $0 \in S^{(i)}$.
\end{proof}

We'll first briefly argue approximate truthfulness; the argument for \TPDD
carries over unchanged. Note that the final output $\tilde{x}^{(i)}$ for any
unsatisfied bidder is $\bar{x}^{(i)}$ in expectation, while any unsatisfied
bidder gets her favorite good at the final prices. Thus, truthfulness is clear
by \Cref{thm:truthful}.

\begin{corollary} \label{thm:tight-truth}
  Suppose we run \TRPDD on a class of packing linear programs with a dual bound
  $\tau$. Then, \TRPDD is $(\rho, \gamma)$-approximately truthful for
  \[
    \rho = e^\eps ,
    \qquad
    \gamma = \alpha(2e^\eps - 1) + \delta \max \{ V, C \tau \sqrt{k} \} .
  \]
\end{corollary}

Next, let's look at the welfare guarantee. We can bound the possible welfare
loss from reassigning unsatisfied bidders by the same argument form \TPDD, and
we can use a standard Chernoff-Hoeffding bound to bound the possible welfare
loss from rounding.

\begin{theorem}[see e.g., \citet{concentration-book}] \label{thm:chernoff}
  Suppose $X = \sum_i X_i$ is a finite sum of of independent, bounded random
  variables $0 \leq X_i \leq M$, and $\mu_L \leq \EE[X] \leq \mu_H$. Then, for
  any $\beta > 0$,
  \[
    \Pr\left[ X > \mu_H + \sqrt{3 M \mu_M \log(1/\beta)} \right] \leq \beta
    \quad \text{and} \quad
    \Pr\left[ X < \mu_L - \sqrt{2 M \mu_L \log(1/\beta)} \right] \leq \beta
  \]
\end{theorem}

\begin{theorem} \label{thm:tight-opt}
  Let $\beta > 0$ be given. Then with probability at least $1 - \beta$, \TRPDD
  run on a packing linear program $\cO = (S, v, c, b)$ with gradient sensitivity
  bounded by $\sigma$, a dual bound $\tau$, and width bounded by $w$ has
  objective satisfying $\< v, \tilde{x} \> \geq \OPT - \alpha$, for
  \[
    \alpha = O \left(
      \kappa \cdot \OPT
      +
      \left( \frac{k \tau \sigma }{\epsilon}
      \left( 1 + \frac{V}{\alpha} \right)
      + \sqrt{V \cdot \OPT} \right)
      \log\left(\frac{wk}{\beta} \right)
      \log^{1/2}\left( \frac{w}{\delta} \right)
    \right) ,
  \]
  for
  \[
    \kappa = \max_j \xi /b_j < 1 ,
  \]
  so
  \[
    \min_j b_j \gg \xi =
    \Omega \left(
      \| b \|^{1/2}_\infty
      + \frac{k \tau \sigma C_\infty}{\epsilon}
      \log^2\left(\frac{w k}{\beta} \right)
      \log^{1/2} \left( \frac{w}{\delta} \right)
      \left( \frac{1}{\tau} + \frac{C_\infty k}{\alpha} \right)
    \right)
    .
  \]
\end{theorem}
\begin{proof}
  Let $\bar{x}$ be the output from the $(\epsilon/2, \delta/2)$-private \PDD on
  the reduced problem $\cO_{red}$.  Note that the objective $\< v, \tilde{x} \>$
  is the sum of $i$ independent random variables, each bounded in $[0, V]$. We
  can lower bound the expected welfare with by
  \Cref{cor:truthful-acc}; reassigning the unsatisfied bidders leads to welfare
  at least
  \[
    \EE \left[ \< v, \tilde{x} \> \right] = \< v, \bar{x} \>
    \geq
    \OPT_{red}
    - \frac{80 \sqrt{8} k \tau \sigma }{\epsilon} \log\left(\frac{4
        w^2 k}{\beta} \right) \log^{1/2} \left( \frac{2w^2}{\delta} \right)
    \left(2 + \frac{V}{\alpha} \right) .
  \]
  with probability at least $1 - \beta/2$. Applying the concentration bound
  (\Cref{thm:chernoff}), we have
  \begin{align*}
    \< v, \tilde{x} \>
    &>
    \OPT_{red}
    - \frac{80 \sqrt{8} k \tau \sigma }{\epsilon} \log\left(\frac{4 w^2
        k}{\beta} \right) \log^{1/2} \left( \frac{2w^2}{\delta} \right)
    \left(2 + \frac{V}{\alpha} \right)
    - \sqrt{2V \cdot \OPT_{red} \log \left( \frac{2}{\beta} \right)} \\
    &\geq (1 - \kappa) \OPT
    - \frac{80 \sqrt{8} k \tau \sigma }{\epsilon} \log\left(\frac{4 w^2
        k}{\beta} \right) \log^{1/2} \left( \frac{2w^2}{\delta} \right)
    \left(2 + \frac{V}{\alpha} \right)
    - \sqrt{2V \cdot \OPT \log \left( \frac{2}{\beta} \right)} \\
    &= \OPT - O \left(
      \kappa \cdot \OPT
      +
      \left( \frac{k \tau \sigma }{\epsilon}
        \left( 1 + \frac{V}{\alpha} \right)
        + \sqrt{V \cdot \OPT} \right)
      \log\left(\frac{wk}{\beta} \right)
      \log^{1/2}\left( \frac{w}{\delta} \right)
    \right)
  \end{align*}
  with probability at least $1 - \beta/2$, so everything holds with probability
  at least $1 - \beta$.
\end{proof}

Finally, let's look at the constraint violation.

\begin{theorem} \label{thm:tight-feas}
  Let $\beta > 0$ be given. Then with probability at least $1 - \beta$, \TRPDD
  run on a packing linear program $\cO = (S, v, c, b)$ with gradient sensitivity
  bounded by $\sigma$, a dual bound $\tau$, and width bounded by $w$ has
  produces an exactly feasible solution as long as
  \[
    \kappa = \max_j \xi /b_j < 1 ,
  \]
  so
  \[
    \min_j b_j \gg \xi =
    \Omega \left(
      \| b \|^{1/2}_\infty
      + \frac{k \tau \sigma C_\infty}{\epsilon}
      \log^2\left(\frac{w k}{\beta} \right)
      \log^{1/2} \left( \frac{w}{\delta} \right)
      \left( \frac{1}{\tau} + \frac{C_\infty k}{\alpha} \right)
    \right)
    .
  \]
\end{theorem}
\begin{proof}
  Let $\bar{x}$ be the output from the $(\epsilon/2, \delta/2)$-private \PDD on
  the reduced problem $\cO_{red}$. For each constraint,note that the objective
  $\< c_j, \tilde{x} \>$ is the sum of $i$ independent random variables, each
  bounded in $[0, C_\infty]$. We can bound the expected left-hand side of each
  constraint for $\tilde{x}$ by \Cref{cor:truthful-acc}; reassigning the
  unsatisfied bidders makes the constraints at most
  \[
    \EE \left[ \< c_j, \tilde{x} \> \right] = \< c_j, \bar{x} \>
    \leq
    b_j - \xi +
    \frac{80 \sqrt{8} k \tau \sigma }{\epsilon} \log\left(\frac{4 w^2
        k^2}{\beta} \right) \log^{1/2} \left( \frac{2w}{\delta} \right) \left(
      \frac{2}{\tau} + \frac{C_\infty k}{\alpha} \right)
  \]
  with probability at least $1 - \beta/2k$, since the total amount any agent
  contributes to the constraints is $C_1 = C_\infty k$. Applying the
  concentration bound (\Cref{thm:chernoff}), we have
  \begin{align*}
    \< c_j, \tilde{x} \>
    &\leq
    b_j - \xi +
    \frac{80 \sqrt{8} k \tau \sigma }{\epsilon} \log\left(\frac{4 w^2
        k^2}{\beta} \right) \log^{1/2} \left( \frac{2w}{\delta} \right) \left(
      \frac{2}{\tau} + \frac{C_\infty k}{\alpha} \right) \\
    &+ \sqrt{3 \left(
        b_j - \xi +
        \frac{80 \sqrt{8} k \tau \sigma }{\epsilon} \log\left(\frac{4 w^2
            k^2}{\beta} \right) \log^{1/2} \left( \frac{2w}{\delta} \right) \left(
          \frac{2}{\tau} + \frac{C_\infty k}{\alpha} \right)
        \right)
      C_\infty \log \left(
        \frac{2k}{\beta} \right)} \\
    &\leq b_j - \xi +
    \sqrt{3 b_j} + \frac{160 \sqrt{8} k \tau \sigma C_\infty}{\epsilon}
    \log^2 \left(\frac{4 w^2 k^2}{\beta} \right) \log^{1/2} \left(
      \frac{2w}{\delta} \right) \left( \frac{2}{\tau} + \frac{C_\infty
        k}{\alpha}
    \right) \leq b_j
  \end{align*}
  with probability at least $1 - \beta/2k$, so taking a union bound over all $k$
  constraints, everything holds with probability at least $1 - \beta$.
\end{proof}

\begin{remark*}
  While we have presented \TRPDD as achieving approximate truthulness, we can
  also run \TRPDD just for the rounding and exact feasibility by letting the
  truthfulness parameter $\alpha$ be large; the welfare and constraint violation
  bounds degrade gracefully.
\end{remark*}

Finally, it is straightforward to show that \TRPDD is private.

\begin{theorem} \label{true-priv}
  \TRPDD satisfies $(\epsilon, \delta)$-joint differential privacy.
\end{theorem}
\begin{proof}
  By the privacy of \PDD (\Cref{thm:dual-priv,thm:primal-jdp}) and the billboard
  lemma (\Cref{billboard}).
\end{proof}

\paragraph*{Comparison with \citet{HHRRW14}}

In recent work, \citet{HHRRW14} give an algorithm for the $d$-demand allocation
problem (a packing linear program) we considered in \Cref{sec:examples}, but
assuming additionally the \emph{gross substitutes condition} \citep{GS99} from
the economics literature on agent valuations. In this setting, \citet{HHRRW14}
give an $(\epsilon, 0)$-joint differentially private algorithm with the
following welfare.

\begin{theorem}[\citet{HHRRW14}] \label{thm:gs-welfare}
  There is an algorithm that on input a $d$-demand problem with gross
  substitutes valuations and goods with supply $s$, with high probability finds
  a solution with welfare at least
  \[
    \OPT - \alpha n ,
  \]
  and exactly meets the supply constraints, as long as
  \[
    s = \tilde{\Omega} \left( \frac{d^3}{\alpha^3\epsilon} \right),
  \]
  ignoring logarithmic factors, for any $\alpha > 0$.
\end{theorem}

While we do not have algorithms specific to the gross substitutes case, we can
consider the $d$-demand problem (a packing linear program) with general
valuations, rather than gross substitutes.  As discussed in \Cref{cor:ddemand},
$\tau = 1$ is a dual bound, $\sigma = d\sqrt{2}$ bounds the gradient
sensitivity, and $w = n$ bounds the width. The maximum welfare for any agent is
$V = 1$ and each agent contributes at most $C_\infty = 1$ towards any coupling
constraint. While we can also achieve approximate truthfulness (unlike
\citet{HHRRW14}), we will take $\alpha$ to be large; this makes \TRPDD
non-truthful. Applying \Cref{thm:tight-opt}, we can lower bound the welfare of
\TRPDD on a $d$-demand problem.

\begin{corollary} \label{cor:tight-d-demand}
  With high probability, \TRPDD run on a $d$-demand problem with supply $s$ for
  each good finds a solution with welfare at least
  \[
    \OPT - \tilde{O} \left(
      \left( \frac{1}{\sqrt{s}} + \frac{kd}{s\epsilon} \right) \cdot \OPT
      + \frac{kd}{\epsilon}
    \right)
  \]
  and exactly meets the supply constraints, as long as $s = \tilde{\Omega}
  (kd/\epsilon)$, ignoring logarithmic factors.
\end{corollary}
\begin{proof}
  Note that the max scalar $\|b\|_\infty$ in the $d$-demand problem is simply
  supply for each item $s$. The welfare theorem for \TRPDD
  (\Cref{thm:tight-opt}) holds as long as
  \[
    \min_j b_j = s = \tilde{\Omega} \left(
      \sqrt{s} + \frac{k\sigma}{\epsilon}
    \right) .
  \]
  As discussed in \Cref{sec:examples}, the gradient sensitivity for the
  $d$-demand problem is bounded by $\sigma = d\sqrt{2}$. So, we need $s \gg
  kd/\epsilon$. The welfare guarantee follows from \Cref{thm:round-opt},
  plugging in parameters $V = C_\infty = 1$ and noting that $\kappa = O \left(
    1/\sqrt{s} + kd/s\eps \right)$.
\end{proof}

The two algorithms are somewhat incomparable, for several reasons:
\begin{itemize}
  \item The welfare of \TRPDD (\Cref{cor:round-d-demand}) depends on $k$ (but
    not $n$), while the welfare guarantee in \Cref{thm:gs-welfare} depends on
    $n$ (but not $k$).
  \item The algorithm of \citet{HHRRW14} requires the gross substitutes condition
    while \TRPDD does not.
  \item The algorithm of \citet{HHRRW14} satisfies pure $(\epsilon, 0)$-joint
    differential privacy, while \TRPDD satisfies $(\epsilon, \delta)$-joint
    differential privacy only for $\delta > 0$.
\end{itemize}

Nevertheless, we can try to make a rough comparison. In the algorithm by
\citet{HHRRW14}, take $\alpha \approx (d^2/k)^{1/3}$, and say the supply $s
\approx d^3/\alpha^3 \epsilon \approx kd/\epsilon$ (the minimum needed for both
\Cref{thm:gs-welfare} and \Cref{cor:round-d-demand} to apply). Then by
\Cref{cor:tight-d-demand}, \TRPDD gets welfare at least
\[
  (1 - \tilde{O}(1)) \cdot \OPT - \tilde{O} \left( \frac{k d}{\epsilon} \right)
\]
versus $\OPT - \alpha n \approx \OPT - n(d^2/k)^{1/3}$. Thus, \TRPDD improves
when $n$ is large compared to $k$ and $\OPT$:
\[
  n \gg \frac{k^{4/3}}{d\epsilon} + \frac{k^{1/3}}{d^{2/3}} \cdot \OPT  .
\]

\SUBSECTION{Exact feasibility, take $2$}
\label{sec:rounddude}

If we don't need approximate truthfulness, we can round and guarantee exact
feasibility in a different way, with an accuracy guarantee incomparable to
\Cref{thm:tight-opt}. We will continue to work with packing linear programs with
a null action (\Cref{def:packing}), assuming one more thing.

\begin{assumption*}
  We will consider classes of packing linear programs $\cO = (S, v, c, b)$ with
  one extra condition: Each $S^{(i)}$ is a polytope such that at each vertex
  $x$, for all $j$, $c^{(i)}_j(x) = 0$ or $c^{(i)}_j(x) \geq L > 0$. This
  parameter $L$ is valid for the entire class; in particular, it does not depend
  on private data.
\end{assumption*}

Our modification to the solution of \PDD will work in two steps. Instead of
tightening the constraints like \TRPDD, we run \PDD on the original linear
program. Similar to \TRPDD, each agent will then round her solution by selecting
a uniformly random best response from her set of best responses $x^{(i)}_1,
\dots, x^{(i)}_T$. To handle the constraint violation, we will maintain a
differentially private flag on each constraint, which is raised when the
constraint goes tight. In order, we will take agent $i$'s rounded solution if
the flag is down for every constraint she contributes to, i.e., for every $j$
with $\< c^{(i)}_j, x^{(i)}\> > 0$.  Otherwise, she goes \emph{unserved}: we
give her solution $0$. The full code is in \Cref{alg:rounding}.

\begin{algorithm}[h]
  \caption{$\RPDD(\cO, \sigma, \tau, w, \eps, \delta, \beta)$}
  \label{alg:rounding}
  \begin{algorithmic}
    \STATE{{\bf Input}: Packing linear program $\cO = (S, v, c, b)$ with
      $n$ agents and $k$ coupling constraints, gradient sensitivity
      bounded by $\sigma$, a dual bound $\tau$, width bounded by $w$, and
      privacy parameters $\eps > 0, \delta \in
      (0, 1/2)$, and confidence parameter $\beta \in (0, 1)$.}

    \STATE{{\bf Initialize}:
      \begin{mathpar}
        \delta' := \frac{\delta}{2} ,
        \and \eps' := \frac{\eps}{2 \sqrt{8T\ln(1/\delta')}} ,
        \and \zeta := \frac{8 (\log n + \log(3k/\beta) )}{\epsilon'} ,
        \and T_j := b_j - \zeta,
        \quad F_j := \mathsf{Sparse}(\eps', T_j) \text{ for } j \in [k] .
      \end{mathpar}
    }

    \STATE{{\bf Run} $\PDD$:
      \[
        (\bar{x}, \bar{\lambda}) := \PDD(\cO, \sigma, \tau, w, \epsilon/2,
        \delta/2, \beta/3) .
      \]
    }
    \STATE{{\bf for each } agent $i = 0 \dots n$:}
    \INDSTATE[1]{Select $\tilde{x}^{(i)}$ uniformly at random from best responses
      $x^{(i)}_1, \dots, x^{(i)}_T$.
    }
    \INDSTATE[1]{{\bf if } $F_j = \top$ for some constraint with
      $c^{(i)}_j(\tilde{x}^{(i)}) > 0$:}
    \INDSTATE[2]{Set $\hat{x}^{(i)} := 0$.}
    \INDSTATE[1]{{\bf else } Set $\hat{x}^{(i)} := \tilde{x}^{(i)}$.}
    \INDSTATE[1]{Query each sparse vector $j$ with
      \[
        q^{(i)}_j := \sum_{l = 0}^n c^{(l)}_j (\hat{x}^{(l)}) .
      \]
    }
    \STATE{{\bf Output}: $\hat{x}^{(i)} \text{ to agents } i \in [n]$.}
  \end{algorithmic}
\end{algorithm}

Let's consider the first step. By almost the same analysis as for \TRPDD, we can
show that the rounding procedure degrades the objective and violates the
constraints by only a small amount (past what was guaranteed by
\Cref{thm:bigthm}).

\begin{theorem} \label{thm:rounding}
  Let $\beta > 0$ be given. Suppose each agent $i$ independently and uniformly
  at random selects $\tilde{x}^{(i)}$ from $x^{(i)}_1, \dots, x^{(i)}_T$. Then
  with probability at least $1 - \beta$,
  \begin{itemize}
    \item the objective satisfies $\< v, \tilde{x} \> \geq \OPT - \alpha$, for
      \[
        \alpha = O \left(
          \left( \frac{k \tau \sigma}{\epsilon} +  \sqrt{V \cdot \OPT} \right)
          \log\left(\frac{wk}{\beta} \right)
          \log^{1/2} \left( \frac{w}{\delta} \right)
        \right) ;
      \]
      and
    \item the total constraint violation is bounded by
      \[
        \sum_{j = 1}^k \left( \< c_j, \tilde{x} \> - b_j \right)_+
        = O \left(
          \sqrt{ k C_\infty \|b\|_1 \log \left( \frac{k}{\beta} \right) }
        \right) ,
      \]
      as long as
      \[
        \| b \|_1 = \Omega \left(
          \frac{k \sigma}{\epsilon} \log^2 \left(\frac{w k}{\beta} \right) \log
          \left( \frac{w}{\delta} \right) \max \left\{ \frac{ \sigma
            }{\epsilon}, 1 \right\}
        \right) .
      \]

  \end{itemize}
\end{theorem}
\begin{proof}
  Let $\bar{x}$ be the output from the $(\epsilon/2, \delta/2)$-private
  \PDD.  Note that the objective $\< v, \tilde{x} \>$ is the sum
  of $i$ independent random variables, each bounded in $[0, V]$. By
  \Cref{thm:bigthm}, we can lower bound the expected objective:
  \[
    \EE \left[ \< v, \tilde{x} \> \right] = \< v, \bar{x} \>
    \geq \OPT -
    \frac{160 \sqrt{8} k \tau \sigma }{\epsilon} \log\left(\frac{4 w^2
        k}{\beta} \right)^2 \log \left( \frac{2}{\delta} \right) ,
  \]
  with probability at least $1 - \beta/2$. Applying the concentration bound
  (\Cref{thm:chernoff}), we have
  \begin{align*}
    \< v, \tilde{x} \> &>
    \OPT -
    \frac{160 \sqrt{8} k \tau \sigma }{\epsilon} \log\left(\frac{4 w^2
        k}{\beta} \right)^2 \log \left( \frac{2}{\delta} \right)
    - \sqrt{2V \cdot \OPT \log \left( \frac{2(k + 1)}{\beta} \right)} \\
    &= \OPT - O \left(
      \left( \frac{k \tau \sigma}{\epsilon} +  \sqrt{V \cdot \OPT} \right)
      \log\left(\frac{wk}{\beta} \right) \log^{1/2}\left( \frac{w}{\delta} \right)
    \right)
  \end{align*}
  with probability at least $1 - \beta/2(k + 1)$.

  For the constraints, define $y_j = (\< c_j \bar{x} \> - b_j)_+$ to be the
  constraint violation for $j$, if any. For each constraint, we can bound the
  expected left-hand side of each constraint for $\tilde{x}$:
  \[
    \EE \left[ \< c_j, \tilde{x} \>  \right]
    =
    \< c_j, \bar{x} \> \leq b_j + y_j .
  \]
  Since $\< c_j, \tilde{x} \>$ is the sum of $i$ independent random variables in
  $[0, C_\infty]$, applying \Cref{thm:chernoff} gives:
  \[
    \< c_j, \tilde{x} \> < b_j + y_j + \sqrt{3 (b_j + y_j) C_\infty \log \left(
        \frac{2(k + 1)}{\beta} \right)}
  \]
  with probability at least $\beta/2(k + 1)$. So, the total constraint violation
  is bounded by
  \begin{align*}
    \sum_{j = 1}^k \left( \< c_j, \tilde{x} \> - b_j \right)_+
    &\leq
    \sum_{j = 1}^k y_j + \sqrt{3 (b_j + y_j) C_\infty \log \left( \frac{2(k +
          1)}{\beta} \right)} \\
    &\leq
    \left( \sum_{j = 1}^k y_j \right)
    + \sqrt{3 k C_\infty
      \log \left( \frac{2(k + 1)}{\beta} \right)
      \left( \|b\|_1 + \sum_{j = 1}^k y_j \right) } ,
  \end{align*}
  where the second step is by Jensen's inequality. The sum of $y_j$ is the total
  constraint violation of $\bar{x}$, which is bounded by \Cref{thm:bigthm}:
  \[
    \sum_{j = 1}^k y_j \leq
    \frac{160 \sqrt{8} k \sigma }{\epsilon} \log\left(\frac{4 w^2
        k}{\beta} \right) \log^{1/2} \left( \frac{2w}{\delta} \right)
    \ll  \| b \|_1 ,
  \]
  where we have assumed that the constraint violation guarantees for \PDD are
  non-trivial.\footnote{%
    If the constraint violation for \PDD itself is already too big, then there
    is no hope for getting non-trivial constraint violation for \RPDD.}
  Then,
  \begin{align*}
    \sum_{j = 1}^k \left( \< c_j, \tilde{x} \> - b_j \right)_+
    &\leq
    \frac{160 \sqrt{8} k \sigma }{\epsilon} \log\left(\frac{4 w^2
        k}{\beta} \right) \log^{1/2} \left( \frac{2w}{\delta} \right)
    + \sqrt{6 k C_\infty \|b\|_1 \log \left( \frac{2(k + 1)}{\beta} \right) } \\
    &= O \left(
      \sqrt{ k C_\infty \|b\|_1 \log \left( \frac{k}{\beta} \right) }
    \right)
  \end{align*}
  with probability at least $1 - \beta/2(k + 1)$; the last step holds since
  we have assumed
  \[
    \frac{k \sigma^2 }{\epsilon^2} \log^2 \left(\frac{w
        k}{\beta} \right) \log \left( \frac{w}{\delta} \right)
    \ll  \| b \|_1 .
  \]
  Taking a union bound over the $k + 1$ Chernoff bounds, everything holds with
  probability at least $1 - \beta/2$. Since \PDD succeeds with
  probability $1 - \beta/2$, the total failure probability is $1 - \beta$.
\end{proof}

Now, let's consider the second step. To maintain the flags on each constraint,
we will use $k$ copies of the \emph{sparse vector mechanism}
\citep{DNRRV09}. This standard mechanism from differential privacy takes
a numeric threshold and a sequence of (possibly adaptively chosen) queries.
Sparse vector outputs $\bot$ while the current query has answer substantially
less than the threshold, and outputs $\top$ and halts when the query has
answer near or exceeding the threshold. The code is in \Cref{alg:sv}.

\begin{algorithm}[h]
  \caption{Sparse vector mechanism $\mathsf{Sparse}(\eps, T)$}
  \label{alg:sv}
  \begin{algorithmic}
    \STATE{{\bf Input}: Privacy parameter $\eps > 0$, threshold $T$, and stream
    of queries $q_1, q_2, \dots$.}
    \STATE{{\bf Initialize}: $\hat{T} := T + \Lap \left( \frac{2}{\epsilon}
      \right)$. }
    \STATE{{\bf for each } query $q_i$:}
    \INDSTATE[1]{Let $y := q_i + \Lap\left( \frac{4}{\epsilon} \right)$.}
    \INDSTATE[1]{{\bf if } $y \geq \hat{T}$:}
    \INDSTATE[2]{{\bf Output } $a_i = \top$, {\bf Halt}.}
    \INDSTATE[1]{{\bf else}}
    \INDSTATE[2]{{\bf Output } $a_i := \bot$.}
  \end{algorithmic}
\end{algorithm}

Our queries will measure how large each constraint is for the users who have
already submitted their solution $x^{(i)}$, and the threshold will be slightly
less than the constraint bound $b_j$. We want to argue two things: (a) sparse
vector doesn't halt while the constraint is at least $\alpha$ away from being
tight (for some $\alpha$ to be specified), and (b) each constraint ends up
feasible.  We will use a standard accuracy result about the sparse vector
mechanism.

\begin{lemma}[see e.g., \citet{DR14} for a proof] \label{lem:sv-acc}
  Let $\alpha > 0, \beta \in (0, 1)$. Say sparse vector on query sequence $q_1,
  \dots, q_k$ and threshold $T$ is \emph{$(\alpha, \beta)$-accurate} if with
  probability at least $1 - \beta$, it outputs $\bot$ while $q_i$ has value at
  most $T - \alpha$, and halts on the first query with value greater than $T +
  \alpha$.  Then, sparse vector with privacy parameter $\epsilon$ is $(\alpha,
  \beta)$-accurate for
  \[
    \alpha = \frac{8 (\log k + \log(2/\beta) )}{\epsilon} .
  \]
\end{lemma}

Now, we are ready to show how much objective \RPDD loses in order
to guarantee exact feasibility.

\begin{theorem} \label{thm:round-opt}
  Let $\beta \in (0, 1)$ be given. With probability at least $1 - \beta$,
  \RPDD run on a packing linear program produces a solution
  exactly satisfying all the constraints, and with objective at least $\OPT -
  \alpha$, for
  \[
    \alpha = O \left(
      \sqrt{V}\log\left(\frac{k}{\beta} \right)
      \left( \sqrt{\OPT } +
        \frac{\sqrt{V}}{L} \left(
          \frac{\log n}{\epsilon} + \sqrt{k C_\infty \| b \|_1}
        \right) \right)
    \right) ,
  \]
  as long as
  \[
    \| b \|_1 = \Omega \left(
      \frac{k \sigma}{\epsilon} \log^2 \left(\frac{w k}{\beta} \right) \log
      \left( \frac{w}{\delta} \right) \max \left\{ \frac{ \sigma
        }{\epsilon}, 1 \right\}
    \right) .
  \]
\end{theorem}
\begin{proof}
  Let $\hat{x}$ be the output. Since we make at most $n$ queries to every flag,
  by \Cref{lem:sv-acc} and a union bound over all $k$ flags, with probability at
  least $1 - \beta/3$, the final left-hand side of each constraint is at most
  \[
    \< c_j, \hat{x} \> \leq T_j +
    \frac{32\sqrt{2} (\log n + \log(3k/\beta) \sqrt{\log(2/\delta)})}{\epsilon}
    = b_j ,
  \]
  so $\hat{x}$ is strictly feasible. Furthermore, each constraint with raised
  flag satisfies
  \[
    \< c_j, \hat{x} \> \geq T_j -
    \frac{32\sqrt{2} (\log n + \log(3k/\beta) \sqrt{\log(2/\delta)})}{\epsilon}
    = b_j - 2 \zeta.
  \]
  Now, an agent is only unserved if she contributes to a violated constraint.
  Since the best-response problem of each agent is to maximize a linear function
  over a polytope $S^{(i)}$, all best-responses are vertices. So, an unserved
  agent contributes at least $L > 0$ to violated constraints, and each unserved
  bidder reduces the total constraint violation by at least $L$.

  By \Cref{thm:rounding}, with probability at least $1 - 2\beta/3$, the total
  constraint violation of $\tilde{x}$ is:
  \[
    \sum_{j = 1}^k \left( \< c_j, \tilde{x} \> - b_j \right)_+
    \leq 3 \sqrt{ k C_\infty \|b\|_1 \log \left( \frac{k}{\beta} \right) } .
  \]
  Now the final output $\hat{x}$ has no constraint violation, and has reduced
  the right-hand side of each constraint by at most $2 \zeta$. So, the number of
  agents who are unserved is at most
  \[
    U \leq \frac{1}{L} \left( 2 \zeta +
      + 3 \sqrt{ k C_\infty \|b\|_1 \log \left( \frac{k}{\beta} \right) }
    \right) .
  \]
  Since each unserved agent contributes at most $V$ to the objective, the final
  output $\hat{x}$ reduces the objective of $\tilde{x}$ by at most $UV$, so
  \begin{align*}
    \< v, \hat{x} \> &\geq \OPT
    - \frac{160 \sqrt{8} k \tau \sigma }{\epsilon} \log\left(\frac{6 w^2
        k}{\beta} \right) \log^{1/2} \left( \frac{2w}{\delta} \right) \\
    &- \sqrt{2V \cdot \OPT \log \left( \frac{3(k + 1)}{\beta} \right)} \\
    &- \frac{V}{L} \left(
      \frac{64\sqrt{2} (\log n + \log(3k/\beta)
        \sqrt{\log(2/\delta)})}{\epsilon}
      + 3 \sqrt{ k C_\infty \|b\|_1 \log \left( \frac{k}{\beta} \right) }
    \right) \\
    &= \OPT - O \left(
      \sqrt{V}\log\left(\frac{k}{\beta} \right)
      \left( \sqrt{\OPT } +
        \frac{\sqrt{V}}{L} \left(
          \frac{\log n}{\epsilon} + \sqrt{k C_\infty \| b \|_1}
        \right) \right)
    \right) .
  \end{align*}
  With probability at least $1 - \beta$, sparse vector is accurate, the rounding
  succeeds, and \PDD succeeds.
\end{proof}

To show privacy, we use a result about the privacy of sparse vector.

\begin{theorem}[see e.g., \citet{DR14} for a proof] \label{thm:sv-priv}
  Let $\epsilon > 0$. $\mathsf{Sparse}(\epsilon, T)$ is
  $\epsilon$-differentially private.
\end{theorem}

Privacy of \RPDD follows directly.

\begin{theorem} \label{thm:rounding-priv}
  Let $\epsilon > 0, \delta \in (0, 1/2)$. Then, \RPDD satisfies
  $(\epsilon, \delta)$-joint differential privacy.
\end{theorem}
\begin{proof}
  By \Cref{thm:sv-priv} and \Cref{lem:composition}, the flags are $(\epsilon/2,
  \delta/2)$-differentially private. \Cref{thm:dual-priv} shows that the result
  of \PDD is $(\epsilon/2, \delta/2)$-jointly differentially private. By the
  billboard lemma (\Cref{billboard}) and composition (\Cref{lem:composition}),
  \RPDD is $(\epsilon, \delta)$-jointly differentially private.
\end{proof}

\paragraph*{Comparison with \citet{HHRRW14}}

Like for \TRPDD, we can compare the welfare guarantee of \RPDD to the welfare
guarantee of \citet{HHRRW14} on the $d$-demand allocation problem.  As discussed
there in \Cref{cor:ddemand}, $\tau = 1$ is a dual bound, $\sigma = d\sqrt{2}$
bounds the gradient sensitivity, and $w = n$ bounds the width. The maximum
welfare for any agent is $V = 1$ and each agent contributes at most $C_1 = 1$
total towards all coupling constraints. Each agent's feasible set is simply the
simplex $\{ x : \RR^m \mid \sum_{i = 1}^m x_i \leq 1 \}$, so the minimum
non-zero coordinate at any vertex is $L = 1$.  Applying \Cref{thm:round-opt}, we
can lower bound the welfare of \RPDD on a $d$-demand problem.

\begin{corollary} \label{cor:round-d-demand}
  With high probability, \RPDD run on a $d$-demand problem with supply $s$ for
  each good finds a solution with welfare at least
  \[
    \OPT - \tilde{O} \left( \frac{k \sqrt{ s }}{\epsilon} \right)
  \]
  and exactly meets the supply constraints,
  as long as $s = \Omega (d^2/\epsilon^2)$, ignoring logarithmic factors.
\end{corollary}
\begin{proof}
  Note that the sum of the scalars $\|b\|_1$ in the $d$-demand problem is simply
  the total number of items $sk$. The welfare
  theorem for \RPDD (\Cref{thm:round-opt}) holds as long as
  \[
    sk = \| b \|_1 \geq \tilde{\Omega} \left(
          \frac{k \sigma}{\epsilon} \max \left\{ \frac{ \sigma }{\epsilon}, 1
          \right\}
    \right) .
  \]
  As discussed in \Cref{sec:examples}, the gradient sensitivity for the
  $d$-demand problem is bounded by $\sigma = d\sqrt{2}$. So, we need $s \geq
  d^2/\epsilon^2$. The welfare guarantee follows from \Cref{thm:round-opt},
  plugging in parameters $V = L = C_1 = 1$ and noting that the maximum welfare
  $\OPT$ is at most the total number of goods $sk$, so $\OPT \leq \| b \|_1$.
\end{proof}

While the welfare guarantees are somewhat incomparable (see discussion for
\TRPDD), we can try to make a rough comparison. Take $\alpha \approx (d
\epsilon)^{1/3}$, and say the supply $s \approx d^3/\alpha^3 \epsilon \approx
d^2/\epsilon^2$ (the minimum needed for both \Cref{thm:gs-welfare} and
\Cref{cor:round-d-demand} to apply). Then by \Cref{cor:round-d-demand}, \RPDD
gets welfare at least
\[
  \OPT - \tilde{O} \left( \frac{k d}{\epsilon^2} \right)
\]
versus $\OPT - \alpha n \approx \OPT - n(d \epsilon)^{1/3}$. Thus, \RPDD
improves for larger $n$:
\[
  n \gg \frac{k d^{2/3}}{\epsilon^{7/3}} .
\]
Unlike \TRPDD, the welfare loss of \RPDD depend on $\OPT$, but does degrade for
larger $s$.
\fi

\paragraph*{Acknowledgments} This paper has benefitted from conversations with
many people. We would like to particularly thank Tim Roughgarden, who made
significant contributions to early discussions about this work during his
sabbatical at Penn. We would also like to thank Moritz Hardt for suggesting
the smart-grid example application, and Jon Ullman for enlightening
conversations.

\ifsubmission
  \shortfalse
  \firstfalse

  \vfill
  \pagebreak

  \appendix
  \setcounter{page}{1}
  
\else

\fi

\bibliographystyle{plainnat}
\bibliography{header,refs}

\ifsubmission
\else
  \appendix
\fi

\section{Private Online Linear Optimization}
\label{sec:polo}

In this section we consider a private version of the
online linear optimization problem. The techniques we use to solve this problem are relatively standard (and are similar to solutions in e.g. \cite{KPRU14} and \cite{BST14}), but we work in a somewhat different setting, so we provide proofs here for completeness. 

 The learner has action set
$\cP\subset \RR^k$, and the adversary has action space $\cX\subset
[-X, X]^k$. Given any action $p$ of the learner, and action $x$ of the
adversary, the (linear) loss function for the learner is $\ell(p, x) =
\langle p, x\rangle$. For any sequence of $T$ actions from the
adversary $\{x_1, \ldots, x_T\}$, the learner's goal is to minimize
the regret defined as
\[
\cR_T = \frac{1}{T} \sum_{t=1}^T \langle p_t, x_t \rangle - \min_{p\in
  \cP} \frac{1}{T} \sum_{t=1}^T \langle p, x_t\rangle.
\]

For privacy reason, the learner only get to observe noisy and private
versions of the adversary's actions. In particular, we can think of
the loss vectors over $T$ rounds as a $(T\times k)$-dimensional
statistics $x = (x_1, \ldots, x_T)$ some underlying sensitive
population $D$. Suppose we add noise sampled from the Gaussian
distribution $\cN(0,\sigma^2)$ on every entry of $x$. We will
determine the scale of $\sigma$ in the privacy analysis, and we
will use the following concentration bound of Gaussian distribution.

\begin{fact}[Gaussian Tails]\label{gaussian}
Let $Y$ be a random variable sampled from the distribution $\cN(0,
\sigma^2)$ and $a = \ln{2} / (2\pi)$, then for all $\lambda > 0$
\[
\Pr[|Y| > \lambda] \leq 2\exp\left( -a \lambda^2/\sigma^2 \right).
\]
\end{fact}

\begin{fact}[Gaussian Sums]\label{gsum}
  Let $Y_1, \ldots, Y_n$ be independent variables with distribution
  $\cN(0, \sigma^2)$. Let $Y = \sum_i Y_i$. Then, the random variable
  $Y\sim \cN(0, n\sigma^2)$, and so
\[
\Pr[|Y| > \lambda] \leq 2\exp\left( -a \lambda^2/(n\sigma^2) \right).
\]
\end{fact}

Let $\{\hat x_t\}$ be the noisy loss vectors. The learner will update
the action $p_t$ using projected gradient descent
\[
p_{t+1} = \Pi_\cP \left[p_t - \eta\hat x_t\right],
\]
where $\Pi_\cP$ is the Euclidean projection map onto the set $\cP$:
$\min_{p'} \|p - p'\|_2$, $\eta$ is the step size.

Before we show the regret bound for our noisy gradient descent, we
here include the no-regret result for standard online gradient descent
(with no noise).

\begin{lemma}[\citep{zinkevich}]
\label{lem:zinkevich}
  For any actions and losses space $\cP$ and $\cX$, the gradient
  descent algorithm: $p_{t+1} = \Pi_\cP \left[p_t - \eta x_t\right]$
  has regret
\[
\cR_T \leq \frac{\|\cP\|^2}{2\eta T} + \frac{\eta \|\cX\|^2}{2}.
\]
\end{lemma}

We are now ready to show the
following regret bound for this noisy gradient descent.

\begin{theorem}
\label{thm:noisy-GD}
Let $\|\cP\| = \max_{p\in\cP} \|p\|$ and $\|\cX\|_2 = \max_{x\in \cX}
\|x\|_2$, then with probability at least $1 - \beta$,
\[
\cR_T = O\left(\frac{\|\cP\|_2\sqrt{k}}{\sqrt{T}}\left( X + \sigma
    \log\left(\frac{Tk}{\beta} \right)\right) \right) .
\]
\end{theorem}

\begin{proof}
Let $\nu_t$ denote the noise vector we have in round $t$, we can
decompose the regret into several parts
\begin{align*}
  \cR_T &= \frac{1}{T}\sum_{t=1}^T \langle p_t, x_t\rangle -
  \frac{1}{T}
  \min_{p\in \cP}\sum_{t=1}^T \langle p , x_t \rangle\\
  &= \frac{1}{T} \sum_{t=1}^T \langle p_t , \hat x_t\rangle -
  \frac{1}{T} \sum_{t=1}^T \langle p_t , \nu_t\rangle -
  \frac{1}{T}\left[ \min_{p\in \cP}\sum_{t=1}^T \langle p, x_t\rangle
    - \min_{\hat p\in \cP}\sum_{t=1}^T \langle \hat p, \hat x_t\rangle
  \right] -\frac{1}{T} \min_{\hat p\in \cP} \sum_{t=1}^T \langle \hat
  p, \hat x_t  \rangle \\
  &= \left[\frac{1}{T} \sum_{t=1}^T \langle p_t , \hat x_t\rangle -
    \frac{1}{T} \min_{\hat p\in \cP} \sum_{t=1}^T \langle \hat p, \hat
    x_t \rangle\right] - \frac{1}{T} \sum_{t=1}^T \langle p_t ,
  \nu_t\rangle - \frac{1}{T}\left[ \min_{p\in \cP}\sum_{t=1}^T \langle
    p, x_t\rangle - \min_{\hat p\in \cP}\sum_{t=1}^T \langle \hat p,
    \hat x_t\rangle \right]\\
  &= \hat\cR_T - \frac{1}{T}\sum_{t=1}^T \langle p_t , \nu_t\rangle -
  \frac{1}{T}\left[ \min_{p\in \cP}\sum_{t=1}^T \langle p, x_t\rangle
    - \min_{\hat p\in \cP}\sum_{t=1}^T \langle \hat p, \hat
    x_t\rangle \right]\\
  &\leq \hat\cR_T - \frac{1}{T}\min_{p\in \cP}\sum_{t=1}^T \langle p ,
  \nu_t\rangle - \frac{1}{T}\left[ \min_{p\in \cP}\sum_{t=1}^T \langle
    p, x_t\rangle - \min_{\hat p\in \cP}\sum_{t=1}^T \langle \hat p,
    \hat x_t\rangle \right] .
\end{align*}

We will bound the three terms separately.
By the no-regret guarantee of online gradient descent in \Cref{lem:zinkevich},
we have the following the regret guarantee w.r.t the noisy losses if we
set $\eta = \frac{\|\cP\|}{\sqrt{T} \|\hat \cX\|}$
\[
\hat \cR_T =\frac{1}{T} \sum_{t=1}^T \langle p_t , \hat x_t \rangle -
\min_{p\in \cP} \frac{1}{T}\sum_{t=1}^T \langle p, \hat x_t \rangle
\leq \frac{\|\cP\|^2}{2\eta T} + \frac{\eta \|\hat\cX\|^2}{2} =
\frac{\|\cP\|\|\hat\cX\|}{\sqrt{T}},
\]
where $\|\cP\|$ and $\|\hat\cX\|$ denote the bound on the $\ell_2$
norm of the vectors $\{p_t\}$ and $\{\hat x_t\}$ respectively.

Recall that for any random variable $Y$ sampled from the Gaussian
distribution $\cN(0, \sigma^2)$, we know from~\Cref{gaussian}
\[
\Pr[|Y| \geq d \cdot \sigma] \leq 2\exp(-ad^2).
\]
For each noise vector $\nu_t$ and any coordinate $i$, we have with
probability except $\beta/Tk$ that $|\nu_t(i)| \leq \sigma\sqrt{\frac{1}{a}\cdot
\log\left(\frac{2Tk}{\beta} \right)}$. By union bound, we have with
probability except $\beta$ that
\[
\max_t\max_i |\nu_t(i)| \leq \sigma\sqrt{\frac{1}{a}\cdot
\log\left(\frac{2Tk}{\beta} \right)}
\quad \mbox{
  and so for all }t,
\quad \|\nu_t\|_2 \leq\sigma\sqrt{\frac{k}{a}\cdot
\log\left(\frac{2Tk}{\beta} \right)}
\]
Since $\hat x_t = x_t + \nu_t$, we know that
\[
\|\hat\cX\| \leq \sqrt{k}\left(X +  \sigma\sqrt{\frac{1}{a}\cdot
\log\left(\frac{2Tk}{\beta} \right)} \right).
\]

Now we bound the second and third term.
From~\Cref{gsum}, we know with probability at least $1-\beta/k$,
for each coordinate $i$,
\[
\sum_t \nu_t(i) \leq \sigma\sqrt{\frac{T}{a}\cdot
  \log\left(\frac{2k}{\beta} \right)} .
\]
By a union bound, we know with probability at least $1-\beta$
\[
\left\|\sum_t \nu_t/T\right\|_\infty \leq
\frac{\sqrt{\frac{\sigma^2}{a}\log(k/\beta)}}{\sqrt{T}} \qquad\mbox{ and
    so } \qquad \left\|\sum_t \nu_t/T\right\|_2 \leq
 \frac{\sqrt{\frac{k\sigma^2}{a}\log(k/\beta)}}{\sqrt{T}} .
\]
Now we could use Holder's inequality to bound the second term
\[
- \frac{1}{T} \min_{p\in\cP} \sum_{t=1}^T \langle p, \nu_t\rangle \leq
\left|\min_{p\in\cP} \langle p, \sum_{t=1}^T \nu_t/T\rangle\right| \leq
\|\cP\|\frac{{\sqrt{k\sigma^2\log(k/\beta)}}}{\sqrt{aT}} .
\]



Let $p^*\in \arg\min_{p\in\cP} \sum_{t=1}^T \langle p, x_t\rangle$ and
$\hat p^*\in \arg\min_{\hat p\in \cP} \sum_{t=1}^T \langle \hat p,
\hat x_t\rangle$. We can rewrite the third term as
\begin{align*}
  - \frac{1}{T} \left[\min_{p\in \cP}\sum_{t=1}^T \langle p,
    x_t\rangle - \min_{p\in \cP}\sum_{t=1}^T \langle \hat p, \hat
    x_t\rangle \right] &= -\langle p^*, \sum_t x_t/T\rangle +
  \langle\hat p^*, \sum_t \hat
  x_t/T\rangle\\
  &= -\langle p^*, \sum_t \hat x_t/T\rangle + \langle \hat p^*, \sum_t
  \hat x_t/T\rangle - \left\langle p^*, \frac{1}{T}\left(\sum_t x_t -
    \sum_t \hat x_t \right)\right\rangle\\
  \mbox{(Holder's inequality)} &\leq \langle p^* , \frac{1}{T} \sum_t
  (\nu_t)\rangle \leq \|\cP\|\frac{{\sqrt{k\sigma^2\log(k/\beta)}}}{\sqrt{aT}} .
\end{align*}

In the end, we have that
\begin{align*}
  \cR_T &\leq \frac{\|\cP\| \sqrt{k}}{\sqrt{T}}\left( {\left(X +
        \sigma\sqrt{\frac{1}{a}\cdot \log\left(\frac{2Tk}{\beta}
          \right)} \right)} + 2\frac{{\sqrt{\sigma^2\log(k/\beta)}}}{\sqrt{a}}\right)\\
  &\leq \frac{\|\cP\| \sqrt{k}}{\sqrt{T}}\left( X +
        2 \sigma\sqrt{\frac{1}{a}\cdot \log\left(\frac{2Tk}{\beta}
          \right)} \right) \\
  &= \frac{\|\cP\| \sqrt{k}}{\sqrt{T}} \left( X +
    \left( \sqrt{\frac{8\pi}{\log 2}} \right) \sigma \sqrt{
          \log\left(\frac{2Tk}{\beta} \right)} \right) \\
  &= O\left(\frac{\|\cP\|\sqrt{k}}{\sqrt{T}}\left( X + \sigma
      \log\left(\frac{Tk}{\beta} \right)\right) \right) .
\end{align*}

\end{proof}

\end{document}